\documentclass[11pt]{article}

\usepackage{fullpage}
\usepackage[T1]{fontenc}
\usepackage{fourier}

\usepackage[style=alphabetic, maxbibnames=99, maxalphanames=4,
minalphanames=4,url=false,isbn=false,doi=false,eprint=false]{biblatex}
\renewbibmacro{in:}{}
\usepackage{amsmath}
\usepackage{amssymb}
\usepackage{bbm}
\usepackage{mathtools}
\usepackage{amsthm}
\usepackage{graphicx}
\usepackage{subcaption}
\usepackage{tikz}

\renewcommand{\vec}[1]{\boldsymbol{#1}}
\newcommand\cE{\mathcal E}
\newcommand\cF{\mathcal F}
\newcommand\cU{\mathcal U}
\newcommand\Atilde{\widetilde{A}}
\newcommand\Ahat{\widehat{A}}
\newcommand\Utilde{\widetilde{\cU}}
\newcommand\Uhat{\widehat{\cU}}
\newcommand\vu{\vec u}
\newcommand\vv{\vec v}
\newcommand\eps{\epsilon}
\newcommand{\deltars}{\delta_{ab}}
\newcommand{\Ctilde}{\widetilde{C}}
\newcommand{\Cmin}{C_{\mathrm{min}}}
\newcommand{\Cstart}{C_{\mathrm{start}}}
\newcommand{\Cstep}{C_{\mathrm{step}}}
\newcommand{\cact}{c_{\mathrm{cons}}}
\newcommand{\epsbase}{\eps_{\mathrm{base}}}
\newcommand{\cadv}{c_{\mathrm{adv}}}
\newcommand{\Tstart}{T_\mathrm{start}}
\newcommand{\Tend}{T_\mathrm{end}}
\newcommand{\Shat}{\widehat{S}}
\newcommand{\what}{\widehat{w}}

\newcommand\norm[1]{\left\|{#1}\right\|}
\newcommand\abs[1]{\left|{#1}\right|}
\DeclareMathOperator{\sign}{sign}
\DeclareMathOperator{\NC}{NC}
\DeclareMathOperator{\Ber}{Ber}

\newcommand\expec{\mathbb{E}}

\usepackage{hyperref}
\usepackage[capitalize]{cleveref}

\newtheorem{definition}{Definition}[section]
\newtheorem{claim}[definition]{Claim}
\newtheorem{remark}[definition]{Remark}
\newtheorem{theorem}[definition]{Theorem}

\newtheorem{lemma}[definition]{Lemma}
\newtheorem{corollary}[definition]{Corollary}

\AddToHook{env/claim/begin}{\crefalias{definition}{claim}}
\AddToHook{env/remark/begin}{\crefalias{definition}{remark}}
\AddToHook{env/theorem/begin}{\crefalias{definition}{theorem}}
\AddToHook{env/proposition/begin}{\crefalias{definition}{proposition}}
\AddToHook{env/lemma/begin}{\crefalias{definition}{lemma}}
\AddToHook{env/corollary/begin}{\crefalias{definition}{corollary}}

\crefname{definition}{definition}{definitions}
\crefformat{definition}{#2definition~#1#3} 
\Crefformat{definition}{#2Definition~#1#3} 

\crefname{claim}{claim}{claims}
\crefformat{claim}{#2claim~#1#3} 
\Crefformat{claim}{#2Claim~#1#3} 

\crefname{remark}{remark}{remarks}
\crefformat{remark}{#2remark~#1#3} 
\Crefformat{remark}{#2Remark~#1#3} 

\crefname{theorem}{theorem}{theorems}
\crefformat{theorem}{#2theorem~#1#3} 
\Crefformat{theorem}{#2Theorem~#1#3} 

\crefname{proposition}{proposition}{propositions}
\crefformat{proposition}{#2proposition~#1#3} 
\Crefformat{proposition}{#2Proposition~#1#3} 

\crefname{lemma}{lemma}{lemmas}
\crefformat{lemma}{#2lemma~#1#3} 
\Crefformat{lemma}{#2Lemma~#1#3} 

\crefname{corollary}{corollary}{corollaries}
\crefformat{corollary}{#2corollary~#1#3} 
\Crefformat{corollary}{#2Corollary~#1#3} 

\title{Geometric opinion exchange polarizes in every dimension}

\author{Abdou Majeed Alidou\footnote{AIMS Rwanda. Email:\url{abdou@aims.edu.gh}.}
\and
Júlia Baligács\footnote{University of Warsaw. Email:\url{jbaligacs@gmail.com}}
\and
Jan Hązła\footnote{AIMS Rwanda. Email:\url{jan.hazla@aims.ac.rw}.
A.M.A.~and J.H.~were
supported by the Alexander von Humboldt Foundation German research chair
funding and associated DAAD projects No. 57610033 and
57761435.}}

\date{}

\begin{document}
\maketitle

\begin{abstract}
A recent line of work studies models of opinion exchange
where agent opinions about $d$ topics are tracked
simultaneously.
The opinions are represented as vectors on the 
unit $(d-1)$-sphere,
and the update rule is based on the overall correlation
between the relevant vectors. The update rule
reflects the assumption of 
biased assimilation, i.e., a pair of opinions is
brought closer together if their correlation is positive
and further apart if the correlation is negative.

This model seems to induce
the polarization of opinions into two antipodal
groups. This is in contrast to many other known models
which tend to achieve consensus. The polarization
property has been recently proved for $d=2$, 
but the general case of $d\ge 3$ remained open.
In this work, we settle the general case,
using a more detailed 
understanding of the model dynamics
and tools from the theory of random processes.
\end{abstract}

\section{Introduction}

Models of belief formation and exchange 
are studied in several scientific disciplines,
including economics,
social sciences, and computer science. The topic is very relevant
to the functioning of a modern society. At the same time,
a given model and its analysis can contain interesting 
mathematics of general interest.

In this work, we focus
on the  model of ``geometric opinion exchange'' introduced
in~\cite{HJMR} and further studied
in~\cite{gaitonde2021polarization,ABH24}.
In this model, agent opinions are tracked simultaneously
for several topics, and accordingly represented as vectors. The
opinions are updated according to a ``geometric'' rule in the sense that 
an update depends on an overall correlation (scalar product)
between a pair of opinions.

More precisely, let 
\(d, n \geq 2\) denote the number of dimensions and the number of agents, respectively. 
We let \([n]\) denote the set \(\{1, 2, \hdots, n\}\) and refer to agents by indices from this set. 
An \emph{opinion} $\vu_i\in\mathbb{R}^d$ of agent $i$ is a $d$-dimensional vector on the unit sphere\footnote{
The unit sphere assumption can be interpreted as a
finite ``budget of attention'', ensuring that an agent cannot
have extreme opinions for all topics.
See~\cite{HJMR} for a discussion.
}, 
in other words satisfying $\|\vu_i\|=1$.
Given $n$ opinions, let us denote them collectively as a \emph{configuration} $\mathcal{U}=(\vu_1,\ldots,\vu_n)$.

Let $\alpha>0$ and $\mathcal{U}^{0}$ be some initial configuration.
We consider the following random process $(\cU^t)_t$:
Given $\cU^t$, choose $(i,j)\in[n]\times [n]$ uniformly at random.
We will call the pair $(i,j)$ an
\emph{interaction} and also
say that agent~\(j\) \emph{influences} the opinion of agent \(i\) at time $t$. The new configuration $\cU^{t+1}$ has the same
opinions as $\cU^t$, except for agent $i$, whose updated
opinion \(\vec u_i^{t+1}\) is given by
\begin{equation}
    \vec u_i^{t+1} = \frac{\vec w}{\norm{\vec w}} , \; \text{where} \;
    \vec w = \vec u_i^t \, + \, \alpha A_{ij}^t \cdot \vec u_j \; ,
    \label{eq:06}
\end{equation}
where $A_{ij}^t=\langle \vec u_i^t,\vec u_j^t\rangle$.

The motivation behind this update rule is the assumption of
\emph{biased assimilation}: If the opinions are positively correlated, i.e.,  $A_{ij}^t>0$, then agent $i$ responds favorably
to persuasion by $j$ and $\vec u_i^{t+1}$ lies on the great
circle of the sphere somewhere between $\vec u_i^t$ and $\vec u_j^t$.
In other words,  the opinion of agent $i$ moves closer towards~$\vec u_j^t$.
On the other hand, if $A_{ij}^t<0$, then agent $i$
responds negatively and moves away from $\vec u_j^t$
and towards $-\vec u_j^t$.

The distinguishing feature of this model is that it seems
to induce the \emph{polarization of opinions}.
That is, over time, each agent's opinion converges to one of a pair
of limiting points $(\vec u^*,-\vec u^*)$.
This behavior contrasts with many well studied, natural
models, which tend to induce \emph{consensus}:
All opinions converge to a single point $\vec u^*$.
For example, convergence to consensus 
(under natural assumptions) is known for
the DeGroot model~\cite{DeGroot}, voter model~\cite{holley1975ergodic}, and  
Bayesian network models~\cite{MST14}, as well as 
many of their variants
and other models, see the discussion in
\cite{HJMR,ABH24}, 
and more generally~\cite{AO11,MT17}.
As polarization can be observed
in many societal settings, it seems interesting to look
for models where it arises in a natural way.

Therefore, it is a natural objective to
characterize the conditions leading to 
polarization of opinions in the model described above.
This is our objective in this paper.
To state our result, we first need to define the notion of polarization:

\begin{definition}\label{def:polarization}
A configuration $\mathcal{U}$ is \emph{polarized} if, for every
$i,j$, either $\vu_i=\vu_j$ or $\vu_i=-\vu_j$.
We say that a sequence of configurations $(\mathcal{U}^{t})_t$
polarizes
if  $\ \lim_{t\to\infty}\mathcal{U}^{t}$ exists and is a polarized configuration
(where convergence is in the standard topology
in $\mathbb{R}^{d}$).
\end{definition}

Note that a consensus configuration
is also polarized according to this definition.
This is addressed just below in
\Cref{rem:polarization-clusters}.

We will show that the process
$(\cU^t)_t$ almost surely polarizes, unless the initial
configuration $\cU^0$ contains a clear obstacle preventing polarization.
For example, consider an initial configuration
where $A_{1i}=0$ for every $i>1$. From~\eqref{eq:06}, it is clear that the opinion of agent $1$ will remain orthogonal to other opinions
for the rest of time. We will prove that an appropriate
generalization of this scenario is 
the only obstacle preventing polarization.

\begin{definition}[Separable configuration]
\label{def:separable}
    A configuration \(\mathcal{U}\) is \emph{separable} when its
    opinions can be partitioned into two nonempty sets \(S\) and \(T\) such that,
    for every opinion \(\vu \in S\) and \(\vv \in T\), it holds \(\vu \perp \vv\).
\end{definition}

We note that $\cU^t$ is separable if and only
if $\cU^{t+1}$ is separable, see Lemma~2.17 in~\cite{ABH24}.

\begin{theorem}\label{thm:main-polarization}
Let $\cU^0$ be an initial configuration which is not separable.
Then, almost surely, $(\cU^t)_t$
polarizes.
\end{theorem}

\begin{remark}
\label{rem:polarization-clusters}
The notion of polarization from \Cref{def:polarization}
is quite strong, with a couple of caveats.
First, it has nothing to say about the speed of convergence.
We leave the analysis of this aspect as 
a direction for further work.

Second, according to \Cref{def:polarization},
a ``consensus configuration'' with all opinions equal
also counts as a polarized configuration. Since
an initial configuration where all opinions
are sufficiently close to each other converges to consensus
(more on that later), this is unavoidable if
we want to prove \Cref{thm:main-polarization} as stated.

On the other hand, let
$\cU^0$ be an initial configuration and $\widetilde{\cU}^0$
be equal to $\cU^0$ except that some $i$-th opinion
satisfies $\vec u_i^0=-\widetilde{\vec u}_i^0$.
Applying~\eqref{eq:06}, it follows that if the same sequence
of interactions is applied to $\cU^t$ and $\widetilde{\cU}^t$,
also at every time $t$ the opinions in
$\cU^t$ and $\widetilde{\cU}^t$ are equal except that
$\vec u_i^t=-\widetilde{\vec u}_i^t$.

Using this symmetry and a concentration bound, 
one can prove
that if an initial configuration $\cU^0$
is drawn randomly i.i.d.~from a distribution
which is symmetric\footnote{
$P$ is symmetric around 0 if
$P(S)=P(-S)$ for every measurable
$S\in\mathbb{S}^{d-1}$, where
$-S=\{-\vec u:\vec u\in S\}$.
} around 0, then, with high probability
(as the number of agents increases),
the agents polarize into two
opposing groups of roughly equal size.
See Section~2.3.4 in \cite{ABH24}
for more details.
\end{remark}

\subsection{Inactive configurations}

While \Cref{thm:main-polarization} is
intuitive, 
its proof is not straightforward and requires
somewhat detailed understanding of the model dynamics.
Let us describe the main challenge that needs to be overcome. Consider a configuration
$\cU=\cU^t$ 
which is not separable, however, for every pair of opinions,
it holds either $|A_{ij}|\approx 0$ or $|A_{ij}|\approx 1$.
Let us say the configuration $\cU^{t+1}$ is obtained
from $\cU^t$ by agent $j_0$ influencing the opinion of agent~$i_0$.
From~\eqref{eq:06}, the opinion of $i_0$ will move
only by a small amount: If $|A_{i_0j_0}|\approx 0$, then this
is clear. On the other hand, $A_{i_0j_0}\approx 1$ means that the distance
between $\vec u_{i_0}^t$ and $\vec u_{j_0}^t$ is small, and
$\vec u_{i_0}^{t+1}$ lies on the arc between these two vectors, in particular,
it will be close to $\vec u_{i_0}^t$.
Furthermore, a symmetric argument applies if~$A_{i_0j_0}\approx -1$.

Therefore, whenever such an ``almost separable'' configuration
is reached, we need to make sure that the random process
continues making progress and does not
``get stuck'' indefinitely in such a state. While it might be intuitive
that such configurations are unstable and the process
must eventually escape, our proof of this is rather involved.
Let us make this more formal by introducing the notion of an \emph{inactive} configuration:

\begin{definition}
\label{def:eps0-eps1-inactive-v2}
Let $\eps_0,\eps_1> 0$.
A configuration $\cU$ is \emph{$(\eps_0,\eps_1)$-inactive} if, for every
pair of opinions, either $|A_{ij}|> 1-\eps_1^2$ or $|A_{ij}|<\eps_0$.
\end{definition}

It is useful to think of an $(\eps_0,\eps_1)$-inactive configuration
as partitioned into ``clusters'' of opinions which are close (up to sign), 
such that all correlations between clusters
are close to zero:
\begin{definition}[Cluster]
\label{def:cluster}
    Let \(\mathcal{U}\) be a configuration.
    A non-empty set \(C \subset [n]\) is a cluster of~\(\mathcal{U}\) if,
    for every \(i, j \in C\), \(\abs{A_{ij}} > 1/2\), and
    for every \(i \in C\), \(j \not\in C\), \(\abs{A_{ij}} < 1/2\).
\end{definition}

Of course if $\cU$ is $(\eps_0,\eps_1)$-inactive,
then if $i,j$ are in the same cluster
it holds $|A_{ij}|>1-\eps_1^2$ and if they
are in distinct clusters it holds $|A_{ij}|<\eps_0$.
It is readily proved that, for sufficiently small $\eps>0$,
an $(\eps,\eps)$-inactive configuration is uniquely partitioned into
at most $d$ clusters (see \Cref{lem:cluster}).
We can now state our technical result rigorously:
\begin{theorem}\label{thm:technical-informal}
Given $n,d,\alpha$, there exist positive constants $\epsbase>\eps_1>\eps$
and a natural number $T$ such that
the following holds:

Let $\cU^0$ be an $(\eps,\eps)$-inactive configuration with clusters
$S_1,\ldots, S_k$. 
Furthermore, assume that there
exist $i,j\in [n]$ with $0<|A_{ij}^0|<\eps$. Then, almost surely,
there exists $t$ such that $\cU^{tT}$ is
not $(\eps,\eps_1)$-inactive.

Furthermore, for the smallest such $t$, $\cU^{tT}$
is $(\epsbase,\epsbase)$-inactive and
has the same clusters
as $\cU^0$ and, 
with probability at least~0.7, satisfies $|A^{tT}_{ij}|>1-\eps_1^2$ for
every $i,j\in S_a$, $a=1,\ldots,k$.
\end{theorem}

Let us discuss some aspects of the statement of \Cref{thm:technical-informal}.
There is an assumption that there exists
a pair of opinions from different clusters with nonzero correlation.
This is necessary to exclude the cases where the clusters are pairwise
orthogonal (in which case the configuration is separable and will never become active\footnote{
When outlining the proof, we might describe a configuration as ``active'' if it is not
$(\eps_0,\eps_1)$-inactive, where the values of $\eps_0,\eps_1$
are not important or implicit from the context.
In the proofs we only use the rigorous notion
of $(\eps_0,\eps_1)$-inactive configurations.
}), as well as when there is only one cluster.
On the other hand, it might be natural to
prove that
the configuration $\cU^{tT}$ is not
$(\eps,\eps)$-inactive, but we show that it is not
$(\eps,\eps_1)$-inactive for some $\eps_1>\eps$. This 
weaker conclusion makes the proof easier, while still allowing to deduce
\Cref{thm:main-polarization} from \Cref{thm:technical-informal}.

Importantly, the conclusion of the theorem is somewhat stronger than
the statement that the configuration ceases to be inactive.
In fact, we prove that the configuration becomes active, and that,
with fixed positive probability, it becomes active
\emph{because of two opinions in different clusters achieving
a noticeable correlation}. This additional property is helpful
for the following reason. Assume that a configuration
becomes active because there are two opinions $i,j$ in the same
cluster $S_a$ with $|A_{ij}|\le 1-\eps_1^2$, however all opinion pairs between
clusters remain almost orthogonal with absolute correlations less than $\eps$.
Then, it seems possible (indeed likely) that the configuration will become
inactive again by moving the opinions in cluster $S_a$ closer together,
while keeping between-cluster correlations small. If this keeps repeating, the process
might become stuck forever with the same cluster structure.

On the other hand, consider
a configuration that becomes active due to $|A_{ij}|\ge\eps$
for two opinions in different clusters. Then, we will see that,
with a noticeable probability, those two clusters can ``collapse''
into one and the next time the process becomes inactive, it will
have a strictly smaller number of clusters.

Intuitively, the unfavorable outcome of a configuration becoming active
because of the inside-cluster correlations seems very unlikely.
However, excluding it rigorously turns out to be quite difficult.

\subsection{Our contribution and previous work}

Models that utilize 
the update rule~\eqref{eq:06} and
other similar rules were introduced
in~\cite{HJMR}. Other works studying
such models include~\cite{gaitonde2021polarization}
and~\cite{ABH24}.
In particular, \cite{ABH24} introduced the particular dynamics
studied in this paper, and posed the question of convergence to
polarization. Then, they proved \Cref{thm:main-polarization}
restricted to $d=2$, and observed that a crucial property
used in the $d=2$ proof does not hold for $d\ge 3$.

Furthermore, \cite{ABH24} observed a partial result for $d\ge 3$: 
If there exists a fixed $\eps>0$ such that
an $(\eps,\eps)$-inactive initial configuration is almost surely escaped,
then $(\cU^t)_t$ almost surely polarizes. More or less,
they proved that \Cref{thm:technical-informal} implies
\Cref{thm:main-polarization}. However, they left open
if \Cref{thm:technical-informal} holds. Our contribution
is answering that question in the positive and proving 
\Cref{thm:technical-informal}. The derivation of \Cref{thm:main-polarization}
from \Cref{thm:technical-informal} is discussed in \Cref{sec:implication}.

\begin{remark}
The framework in~\cite{ABH24} is more general in that
it discusses update rules of the form
\begin{align}
\vec u_i^{t+1} = \frac{\vec w}{\norm{\vec w}}\;,\qquad
\vec w=\vec u_i^t+f(A_{ij}^t)\cdot\vec u_j 
\end{align}
for a more general class of functions $f:[-1,1]\to\mathbb{R}$
(which they call \emph{stable} functions).
As can be seen in~\eqref{eq:06}, we restrict ourselves to the choice
$f(x)=\alpha\cdot x$. This restriction is mostly for the sake of
concreteness and readability. We do not claim a general proof,
but we do not expect
significant changes in a more general setting.

However, another problem might be worth
of further study. As explained in~\Cref{rem:polarization-clusters},
even though our definition of polarization includes consensus,
if the initial opinions are i.i.d.~and symmetric, then
the ``balanced'' polarization occurs with two opposing groups
of similar size. However, the argument to justify this
works only if the function $f$ satisfies $f(-x)=-f(x)$.

It remains open to understand the group sizes of the two groups 
for general update
functions. One example is
\begin{align}
f(x)=\alpha x\cdot  1[x\ge 0]+\beta x\cdot 1[x<0]\;
\end{align}
for some $\alpha\ne\beta$. For $\alpha>\beta$, this could represent
a scenario where ``positive'' interactions influence agents more strongly than ``negative'' ones.
\end{remark}

\begin{remark}
Furthermore, in the results in~\cite{ABH24},
the pair of agents $(i,j)$ is not necessarily chosen uniformly,
but rather from a fully supported distribution $\mathcal{D}$.
Following our proof, it should be clear that it can be
adapted to all fully supported distributions. (Some of the constants
will have additional dependence on $\min_{i,j}\mathcal{D}(i,j)$).

However, our proof does not apply for
distributions which are not fully supported
(for example, if the agents can influence each other only
along edges of a social network). This is another natural direction
for further work.
\end{remark}

\section{Outline of the proof}
\label{sec:theorem-statement}

We start with a couple of
observations about clusters and inactive configurations.

\begin{lemma}[Lemma 2.20 in \cite{ABH24}]
\label{lem:cluster}
    Let \(\mathcal{U}\) 
    be $(\eps_0,\eps_1)$-inactive for
    $\max(\eps_0,\eps^2_1)\le \frac{1}{256}$.
    Then, the clusters of $\mathcal{U}$ form a partition
    of the set of agents $[n]$. 
    Furthermore, if $\max(\eps_0,\eps^2_1)<\frac{1}{d(d+1)}$, then 
    the number of clusters is at most $d$.
\end{lemma}

\begin{claim}[Lemma 2.8 in~\cite{ABH24}]
\label{cl:consistent-signs}
Let $\eps^2_1< 1/4$.
If $\min(|A_{ij}|,|A_{i\ell}|)\ge 1-\eps_1^2$, then
$|A_{j\ell}|\ge 1-(2\eps_1)^2$ and
$\sign(A_{ij})=\sign(A_{i\ell})\sign(A_{j\ell})$.
In particular, $\sign(A_{ij})=\sign(A_{i\ell})\sign(A_{j\ell})$
whenever $i,j,\ell$ all lie in the same cluster
of an $(\eps_0,\eps_1)$-inactive configuration for
$\eps_1^2<1/4$.
\end{claim}

Next, we observe that one interaction
in a ``sufficiently inactive''
configuration does not change its
clusters.

\begin{lemma}
    \label{cor:bounded_step}
    Let $\cU$ be $(\eps_0,\eps_1)$-inactive
    with $\max(\eps_0,\eps_1^2)\le\frac{1}{4(2+\alpha)^2}$,
    and $\cU'$ reachable from $\cU$ in one step.
    Then, for every pair of agents \((i,j)\):
    \begin{itemize}
        \item if \(\abs{A_{ij}} < \epsilon_0\), then \(\abs{A_{ij}'} < 1/2\).
        \item if \(\abs{A_{ij}} > 1-\epsilon^2_1\), then \(\sign (A_{ij}') = \sign (A_{ij})\) and \(\abs{A_{ij}'} > 1/2\).
    \end{itemize}
In particular, $\cU$ and $\cU'$ have the same clusters
and $\sign (A'_{ij})=\sign (A_{ij})$ for every $i,j$
with $|A_{ij}|>1-\eps_1^2$.
\end{lemma}

\Cref{cor:bounded_step} is proved in \Cref{app:bounded-step}.
Let $\epsbase$ be such that all the results stated above hold, i.e., 
$\epsbase=\min\left(\frac{1}{256},\frac{1}{2d(d+1)},
\frac{1}{4(2+\alpha)^2}\right)$. 
Accordingly, \Cref{lem:cluster}, \Cref{cl:consistent-signs}, and \Cref{cor:bounded_step}
apply to all 
$(\epsbase,\epsbase)$-inactive configurations.
Furthermore, whenever we will be discussing
$(\eps_0,\eps_1)$-inactive configurations, we will always
be assuming $\max(\eps_0,\eps_1^2)\le\epsbase$, even if this is 
not stated explicitly.

Note that $\epsbase$ depends on $d$ and $\alpha$.
In the following, all constants,
as well as implicit constants in the big O notation
are allowed to depend on $n,d,\alpha$.

\subsection{Plan of the proof of \Cref{thm:technical-informal}}
\label{sec:plan}

Let $\cU$ be an $(\eps_0,\eps_1)$-inactive
configuration with clusters $S_1,\ldots, S_k$.
We let
\begin{align}
\label{eq:def-delta}
\delta_0(\cU)&=
\max_{\substack{i\in S_a,j\in S_b\\1\le a< b\le k}} |A_{ij}|\;,
&\delta_1(\cU)&=
\max_{\substack{i,j\in S_a\\ 1\le a\le k}} 
\sqrt{1-|A_{ij}|}\;.
\end{align}
Furthermore, let
\begin{align}
\label{eq:def-q}
Q_0(\cU)&=-\log\delta_0(\cU)\;,&
Q_1(\cU)&=-\log\delta_1(\cU)\;.
\end{align}
So, $\cU$ is $(\eps_0,\eps_1)$-inactive if and only if
$\delta_0(\cU)<\eps_0$ and $\delta_1(\cU)<\eps_1$ or,
equivalently, if $Q_0(\cU)>-\log\eps_0$ and $Q_1(\cU)>-\log\eps_1$.
Furthermore, there exist $i,j$ such that 
$0<|A_{ij}|<\eps_0$ if and only if $Q_0(\cU)<\infty$.
On the other hand,
$Q_1(\cU)\in (0,\infty]$, but the fact that it can
be infinite will not cause any problems.
(Intuitively, $Q_1(\cU^t)=\infty$ is good as we want
to show that there exists a time when $Q_0(\cU^t)\le-\log\eps$ and $Q_1(\cU^t)>-\log\eps_1$.)

Given an initial configuration $\cU^0$, we 
define a random process 
$\delta_0(t)=\delta_0(\cU^t)$ and similarly for
$\delta_1,Q_0,Q_1$. We can now restate \Cref{thm:technical-informal}
using the new notation. It should be clear that the following statement
implies \Cref{thm:technical-informal}:

\begin{theorem}
\label{thm:inactive-unstable-restated}
Given $n,d,\alpha$, there exist some
$\eps_1>\eps>0$ and a natural number $T$ 
such that the following holds:

Let $\cU^0$ be an $(\eps,\eps)$-inactive 
configuration, i.e., it satisfies
$Q_0(0)>-\log\eps$ and $Q_1(0)>-\log\eps$.
Furthemore, assume that $Q_0(0)<\infty$.

Then, almost surely there exists the smallest
nonnegative integer $t$ such that
the configuration remains 
$(\epsbase,\epsbase)$-inactive until time $tT$ and either
$Q_0(tT)\le-\log\eps$ or
$Q_1(tT)\le-\log\eps_1$.
Furthermore, with probability
at least 0.7, it holds
$Q_1(tT)>-\log\eps_1$.
\end{theorem}

Note that, in \Cref{thm:inactive-unstable-restated},
we state that the configuration remains
$(\epsbase,\epsbase)$-inactive over the whole time
from $0$ up to and including $tT$. 
By the previous considerations, that
implies that the clusters remain the same over that
time, and in particular that $Q_0$ and $Q_1$ are always
well-defined with respect to the same set of clusters.

How should we go about proving \Cref{thm:inactive-unstable-restated}?
As a first try, one could hope that there exists some fixed $K$ such that, 
for every $\cU^0$ which is 
$(\eps,\eps)$-inactive, there is a sequence of $K$ interactions such that
$\cU^K$ becomes active. If that holds, then at every step,
independently of the past, we would have a constant positive
probability of becoming active in the following $K$ steps.
That easily implies that the sequence~$(\cU^t)_t$ becomes active
almost surely. 

Perhaps surprisingly, such a property can be proved
in the case of $d=2$. However, for $d\ge 3$ it is false, 
that is, for every $\eps>0$ and $K$,
there exists an $(\eps,\eps)$-inactive configuration that requires more than $K$ steps
to become active. Both of these facts are discussed in more detail in~\cite{ABH24}.

With this optimistic approach having failed, 
it is natural to turn to potential functions.
For example, we can consider
\begin{align}
\delta'(t)&=\sum_{i,j=1}^n \left(A_{ij}^t\right)^2\;.
\end{align}
It is easy to check
that $\delta'(t)\le n^2$ with the equality achieved exactly for polarized
configurations.
While we are not aware of a proof, empirically it appears that
\begin{align}
\label{eq:21}
\expec [\delta'(t+1)\;\vert\;\cU^t]\ge\delta'(t) 
\end{align}
holds for every configuration. 
If that is true, one might hope that
$\lim_{t\to\infty} \delta'(t)=n^2$ holds almost surely, which implies
polarization. However, of course
$|\delta'(t+1)-\delta'(t)|$ can (and will) be arbitrarily small
for inactive configurations. Therefore, even if we proved~\eqref{eq:21},
it is not clear that it would be sufficient for our purposes\footnote{
As an illustration, consider the following simple example. Let $X_0=1/2$
and $X_{t+1}=X_t+B_t\cdot\min(X_t,1-X_t)/2$, where $(B_t)_t$ is i.i.d.~sequence
uniform in $\{-1,1\}$. Since $(X_t)_t$ is a bounded martingale, by a standard application
of martingale theory,
there exists $X$ such that $X=\lim_{t\to\infty} X_t$
holds almost surely, 
and
furthermore $X\sim\Ber(1/2)$.

Now take $Y_t=X_t^2$. By a simple calculation, it holds $\expec[Y_{t+1}\;\vert\;Y_t]>Y_t$.
One might optimistically hope
that almost surely $\liminf_{t}Y_t\ge Y_0=1/4$.
However, this
is contradicted by the fact that
$\Pr[\lim_t Y_t=0]=\Pr[\lim_t X_t=0]=1/2$.
}.

That is the reason for ``taking the logs'' and tracking the quantities $Q_0(t)$
and $Q_1(t)$. Then, we can hope for these random processes to behave
in a comparable way to
random walks with a bias
bounded away from zero. For example, as we want the between-cluster correlations
to increase, ideally we would like $\expec[Q_0(t+1)\;\vert\;\cU^t]\le Q_0(t)-c$
for some constant $c>0$. However, the situation is not so simple,
and it is not hard to find examples where $\expec[Q_0(t+1)\;\vert\;\cU^t]>Q_0(t)$.

A natural workaround to this problem is to hope
that the random process behaves more smoothly
over longer timescales. Accordingly, we can try
showing that
\begin{align}
\label{eq:22}
\expec[Q_0(t+T)\;\vert\;\cU^t]\le Q_0(t)-c
\end{align}
holds for some large (but fixed) value of $T$.
Indeed, with a considerable effort,
we establish such a property.

To understand why~\eqref{eq:22} holds, it is instructive to consider
an inactive configuration where all clusters consist of only one opinion.
In that case, it is possible to show~\eqref{eq:22} by 
implementing the
following sketch:
Elementary calculations show that any interaction where $j$ influences $i$
increases their absolute correlation from $|A_{ij}|$ to at least 
$(1+c)|A_{ij}|$ for some 
fixed $c>0$. Hence, $-\log|A_{ij}|$ decreases by at least
$\log(1+c)$.
On the other hand, as all other opinions are almost orthogonal to $i$
and $j$, it can be established that 
any other correlation $A_{i\ell}$ changes by at most
$O(\delta_0(t)^2)$. For large enough $T$, with high probability, 
the pair of opinions that realizes $\delta_0(t)$ will interact
at least once, and therefore $\delta_0(t+T)\ge (1+c/2)\delta_0(t)$
(where $c/2$ accounts for the 
$O(\delta_0^2)$ factors)
and 
$Q_0(t+T)\le Q_0(t)-\log(1+c/2)$.

However, the situation can be
more complicated for configurations
with larger clusters. For instance, 
if $i,j$ are in one cluster
and $\ell$ in another, with, say, $A_{i\ell}\approx \eps$ and $A_{j\ell}\approx-\eps$,
then the effects of $i$ influencing~$\ell$ and,
subsequently, $j$ influencing $\ell$
may ``cancel out''. Furthermore, interactions between $i$
and~$j$ will bring them closer together, which might have
incidental effect of decreasing $\delta_0$, equivalently
increasing~$Q_0$. A direct analysis of a general situation
seems  complicated.

Instead, we propose the following notion: Consider an inactive configuration
and two of its clusters $S_a,S_b$.
We call the configuration \emph{$(a,b)$-consistent}
if, for every $i,i'\in S_a$ and $j,j'\in S_b$, it holds
\begin{align}
\label{eq:23}
\sign(A_{i'j'})=\sign(A_{ii'})\sign(A_{ij})\sign(A_{jj'})\;.
\end{align}
For example, a configuration where 
$A_{ij}>0$ for every $i,j\in S_a\cup S_b$ is $(a,b)$-consistent.
A consistent configuration has the property that
all interactions between $S_a$ and $S_b$, as well as inside
$S_a$ and $S_b$, tend to increase (or at least not decrease)
the absolute correlations between $S_a$ and $S_b$.
In that sense, the notion of consistency is a useful
generalization of the cluster-size-one scenario.

In \Cref{sec:q0-expectation}, we prove the following useful properties of consistent configurations. At some time 
$t$, let $S_a$ and $S_b$
be the clusters realizing $\delta_0(t)$, i.e.,
$\delta_0(t)=\max_{i\in S_a,j\in S_b}|A_{ij}^t|$.
First, perhaps surprisingly, a careful argument
shows that there is a fixed $K$
such that, for any inactive configuration, there exists
a sequence of $K$ interactions which makes it 
$(a,b)$-consistent. Therefore, an inactive configuration
becomes $(a,b)$-consistent in $O(1)$ steps in expectation.
Second, for some choice of $T=T_0$, \eqref{eq:22} holds for  
every $(a,b)$-consistent configuration.
This is proved using similar ideas as described above
for the cluster-size-one scenario, and using~\eqref{eq:23}. In particular, the only interactions
that can decrease absolute correlations between $S_a$
and $S_b$ must involve a third cluster, and they can
change $\delta_0(t)$ only by $O(\delta_0(t)^2)$.
That leads to the third useful property:
Once an inactive configuration becomes $(a,b)$-consistent,
it must remain so for a long time, essentially
$\Omega(-\log \delta_0(t))$.

These three properties imply~\eqref{eq:22}
for any $(\eps,\eps)$-inactive configuration,
by taking $T$ to be a large multiple of $K+T_0$:
A configuration will have enough time to become
consistent with high probability, 
and once it becomes consistent
it remains
consistent, ``accumulating'' the bias
every $T_0$ steps.

\medskip

The bias property in~\eqref{eq:22},
together with the fact that
$|Q_0(t+1)-Q_0(t)|\le O(1)$,
suffices to conclude that, almost surely,
starting from an $(\eps,\eps)$-inactive configuration,
eventually it holds $\delta_0(t)\ge\eps$
(equivalently $Q_0(t)\le-\log\eps$).
However, recall that we wish to show more than that: We want $Q_0(t)\le-\log\eps$
to occur \emph{before} $Q_1(t)\le-\log\eps_1$,
with constant probability.
Again, we can start from an optimistic hypothesis and try showing
\begin{align}
\label{eq:24}
\expec[Q_1(t+T)\;\vert\;\cU^t]\ge Q_1(t)+c
\end{align}
(or more generally $\expec[Q_1(t+T)-Q_0(t+T)\;\vert\;\cU^t]>Q_1(t)-Q_0(t)+c$).
However, yet again this is false: One can even construct
pathological examples where $Q_1(t)=\infty$ and $Q_1(t+1)<\infty$,
so seemingly there is no control at all over the magnitude
of change of $Q_1(t)$.

Analyzing the problem, it turns out that the cases
with large change in $Q_1(t)$
arise only when $\delta_1(t)\ll \delta_0(t)$,
equivalently $Q_1(t)\gg Q_0(t)$. When $\delta_1(t)$ 
is much smaller than $\delta_0(t)$, then
an interaction between two clusters can induce a (relatively)
large change in $\delta_1(t)$, in other words the change
in $Q_1(t)$ can be determined more by the between-cluster
correlations than the inside-cluster correlations.
To make this observation precise, we 
establish the bound
$Q_1(t+1)\ge\min(Q_0(t),Q_1(t))-O(1)$.

We then show that \eqref{eq:24} holds
in the opposite case when $\delta_1(t)$
is sufficiently larger than $\delta_0(t)$, i.e.,
$Q_1(t)\le Q_0(t)-C$ for a certain fixed $C>0$.
Intuitively, this is because for $\delta_1(t)\gg\delta_0(t)$,
the interactions between clusters can induce only
(relatively) small change in $\delta_1$,
and the interactions inside clusters only
make the clusters ``tighter'', decreasing $\delta_1$
and increasing $Q_1$.

As a result, we can try the following strategy:
Choose $\eps_1> \eps$ such that
$-\log\eps_1\ll -\log\eps-C$. Then, as long as $Q_0(t)\ge-\log\eps$,
whenever $Q_1(t)$ becomes somewhat close to the threshold
$-\log\eps_1$, the condition $Q_1(t)\le Q_0(t)-C$ holds.
Therefore, \eqref{eq:24} applies and we can hope that
the positive bias will tend to prevent $Q_1(t)$ from crossing
$-\log\eps_1$.

\medskip

All in all, our proof is divided into two parts.
First, in \Cref{sec:q0-q1-properties} and \Cref{sec:q0-expectation}
we show that the random processes $Q_0$ and $Q_1$ satisfy the properties explained
above. This is summed up in the following lemma:
\begin{lemma}
\label{lem:p0-p1-properties}
There exist $\eps>0$, $C\ge 1$ and positive integer $T$ such that:

Let $P_0(t)=Q_0(tT)$ and $P_1(t)=Q_1(tT)$. Whenever the configuration
$\cU^{tT}$ is $(\eps,\eps)$-inactive
with $P_0(t)<\infty$, then
$\cU^{t'}$ remains $(\epsbase,\epsbase)$-inactive
for $tT\le t'\le (t+1)T$ and
\begin{align}
|P_0(t+1)-P_0(t)|&\le C\;,
\label{eq:p0-bounded}\\
\expec\left[P_0(t+1)\;\vert\;\mathcal{U}^{tT}\right]
&\le P_0(t)-1/C
\label{eq:p0-martingale}\;.\\
P_1(t+1)&\ge\min(P_0(t),P_1(t))-C\;.\label{eq:p1-lower-bound}
\end{align}
Furthermore, if $P_1(t)\le P_0(t)-C$, then
\begin{align}
|P_1(t+1)-P_1(t)|&\le C\;,
\label{eq:p1-bounded}\\
\expec[ P_1(t+1)\;\vert\;\cU^{tT}]&\ge P_1(t)+1/C\;.
\label{eq:p1-martingale}
\end{align}
\end{lemma}

Then, in \Cref{sec:many-steps-at-once} and \Cref{sec:stochastic-proof}, we show
a more general result:
Any random process that satisfies the conditions of
\Cref{lem:p0-p1-properties}
also satisfies the conclusion of \Cref{thm:inactive-unstable-restated}, i.e.,
$P_0(t)\le-\log\eps$ will occur before
$P_1(t)\le-\log\eps_1$, with probability at least 0.7.

\begin{theorem}\label{prop:two-chains}
\label{thm:random-process}
Let $C>0$ and $\Cmin\in\mathbb{R}$.
Let $(P_0(t),P_1(t))_t$ be a random process adapted to a filtration
$(\cF_t)_t$, with $P_0(t)\in\mathbb{R}$ and $P_1(t)\in\mathbb{R}\cup\{\infty\}$.
Assume that $P_0$ and $P_1$ satisfy
\eqref{eq:p0-bounded}--\eqref{eq:p1-martingale}
(with the constant $C$)
whenever $P_0(t),P_1(t)>\Cmin$.

Then, there exists $\Ctilde=\Ctilde(C)$ such that
the following holds. Let $\Cstart=\Cmin+\Ctilde$.
Assume that $\min(P_0(0),\allowbreak P_1(0))>\Cstart$
and let $t_0=\min\{ t: P_0(t)\le \Cstart\text{ or }
P_1(t)\le \Cmin\}$.
Then, almost surely, $t_0$ is finite,
and furthermore,
\begin{align}
\label{eq:76}
\Pr[P_1(t_0)\le \Cmin]
&\le 0.3\;.
\end{align}
\end{theorem}

Proving \Cref{thm:random-process} also needs some care.
Whenever the process $P_1(t)$ becomes ``dangerously close'' to its threshold
$\Cmin$, it satisfies $P_1(t)\le P_0(t)-C$
and therefore~\eqref{eq:p1-martingale},
so, in any particular instance, the probability that $P_1$ will
cross the threshold is low. However, this might not be enough
if $P_1$ gets ``too many chances'' to cross the threshold. 

To prove \Cref{thm:random-process},
first, we argue that the problem can be reduced to the case where $P_1(0)>P_0(0)-C$.
If $P_1(0)>P_0(0)-C$, then, choosing sufficiently large $\Ctilde$,
the event $P_1(t)\le P_0(t)-C$ must occur at least once
before $P_1(t)$ crosses $\Cmin$.
For a nonnegative integer $\ell$,
let $N_\ell$ be the number of time steps
such that $\Cstart+\ell<P_0(t)\le \Cstart+\ell+1$.
Using~\eqref{eq:p0-bounded} and~\eqref{eq:p0-martingale},
we can apply Azuma's inequality and the union bound
to show that, with good probability, 
the bound $N_\ell\le O(\ell)$ holds for all $\ell$ simultaneously.

As mentioned, before crossing $\Cmin$, the process $P_1(t)$ must satisfy
$P_1(t)\le P_0(t)-C$ for some 
directly preceding contiguous time segment. 
If this time segment starts when
$\Cstart+\ell< P_0(t)\le\Cstart+\ell+1$, then its
duration must be
at least $\Omega(\ell)$ steps. During each of those steps, the condition
$P_1(t)\le P_0(t)-C$ is satisfied, and therefore
~\eqref{eq:p1-bounded}
and~\eqref{eq:p1-martingale} hold. By another application of Azuma,
any such specific segment has a probability of occurring which is
exponentially small in $\ell$. This allows to conclude the proof
by the union bound.
We develop this argument precisely in
\Cref{sec:stochastic-proof}.

\subsection{\Cref{thm:technical-informal} implies
\Cref{thm:main-polarization}}
\label{sec:implication}

Having proved \Cref{thm:technical-informal}, let
us sketch how to deduce \Cref{thm:main-polarization}.
Let $\eps$
and $\eps_1$ be as in \Cref{thm:technical-informal}.
First, if a configuration is $(\eps,\eps)$-inactive
with one cluster, then all opinions are
close to each other, up to minus signs. It is not hard
to deduce that such a configuration polarizes almost surely.

By \Cref{thm:main-polarization},
a configuration which is $(\eps,\eps)$-inactive with
at least two clusters 
(and not separable), eventually becomes active again.
On the other hand,
whenever a configuration is \emph{not} $(\eps,\eps)$-inactive,
then there exists a sequence of $K$ interactions (for some fixed $K$)
that make it $(\eps,\eps)$-inactive. This is because as long as
there exist two opinions with $\eps\le |A_{ij}|\le 1-\eps$,
they can become $\eps$-close in $O(1)$ number of interactions
between them. 

Therefore, we can divide the time into ``epochs'' where in
each epoch the configuration remains inactive with the same clusters.
Let $\NC(\ell)$ be the number of clusters in the $\ell$-th epoch.
From \Cref{lem:cluster}, it holds $1\le\NC(\ell)\le d$, and we want
to show that almost surely $\NC(\ell)=1$ happens for some $\ell$.
However, this can be proved using the second part
of \Cref{thm:technical-informal}: It can be shown that there
exists a fixed $p>0$ such that
$\Pr[\NC(\ell+1)<\NC(\ell)\;\vert\;\NC(\ell)]\ge p$ as long as $\NC(\ell)>1$.

Essentially, this plan has been executed, and the implication from
\Cref{thm:technical-informal} to \Cref{thm:main-polarization} proved,
in~\cite[Theorem 1.7]{ABH24}. The version needed here differs only in details\footnote{\cite{ABH24} considers only $(\eps,\eps)$-inactive configurations, while
in \Cref{thm:technical-informal} the configuration stops being
$(\eps,\eps_1)$-inactive for $\eps_1>\eps$. Furthermore, \cite{ABH24} considers
the first time $t$ such that the configuration is active, while we take
the first time $tT$ for some fixed $T$.}.
The modifications required to handle
these differences are not significant, however, for the sake of correctness, we include a (mostly) self-contained proof in
\Cref{app:implication}.

\medskip

The rest of the paper is dedicated to proving
\Cref{thm:technical-informal}, following the outline
described in \Cref{sec:plan}

\section{Properties of $Q_0$ and $Q_1$}
\label{sec:q0-q1-properties}

Several times, we will be using the following formula which is
easy to check directly.
If $\cU'$ is obtained from $\cU$ by agent $\ell$ influencing~$i$, 
then, for any agent $j$, its new correlation with $i$ is given by
\begin{align}\label{eq:new-aij}
A'_{ij}&=\frac{A_{ij}+\alpha A_{i\ell}A_{j\ell}}{\sqrt{1+(2\alpha+\alpha^2)A_{i\ell}^2}},
&\text{in particular}&
&A'_{i\ell}&=\frac{(1+\alpha)A_{i\ell}}{\sqrt{1+(2\alpha+\alpha^2)A_{i\ell}^2}}
\;.
\end{align}

Recall our plan from \Cref{sec:plan}. For an $(\epsbase,\epsbase)$-inactive configuration,
recall the values $\delta_0,\delta_1,\allowbreak Q_0,Q_1$ defined 
in~\eqref{eq:def-delta} and~\eqref{eq:def-q}. Our objective
is to prove the properties stated in
\eqref{eq:p0-bounded}--\eqref{eq:p1-martingale}.
To do that, we first establish analogous properties
for one step of $Q_0(t)$ and $Q_1(t)$.
We proceed to do so in this section,
with the exception of~\eqref{eq:22},
which is deferred to \Cref{sec:q0-expectation}.

\begin{lemma}
\label{lem:q0-bounded}
If $\cU^t$ is $(\epsbase,\epsbase)$-inactive with $Q_0(t)<\infty$, then
$\frac{\delta_0(t)}{2(1+\alpha)}\le
\delta_0(t+1)\le (1+\alpha)\delta_0(t)$.
In particular, $|Q_0(t+1)-Q_0(t)|\le O(1)$.
\end{lemma}

\begin{proof}
Let $\cU=\cU^t$ and $\cU'=\cU^{t+1}$. By \Cref{cor:bounded_step}, the clusters
of $\cU'$ remain the same.
Let $(i,\ell)$ denote the chosen pair at time $t$, i.e., agent $\ell$ influences agent $i$.
For the upper bound, let $j$ denote an arbitrary agent which is
in a different cluster than agent $i$.
Then $|A_{ij}|\le \delta_0(t)$ and
$\min(|A_{i\ell}|,|A_{j\ell}|)\le\delta_0(t)$.
By \eqref{eq:new-aij}, it follows $|A'_{ij}|\le (1+\alpha)\delta_0(t)$.
Since this holds for every $j$ in a different cluster,
and since $i$ is the only agent whose opinion changed,
it follows
$\delta_0(t+1)\le(1+\alpha)\delta_0(t)$.

For the lower bound, if $\delta_0(t)=|A_{i'j}|$ such 
that $i\notin \{i',j\}$, then $\delta_0(t+1)\ge \delta_0(t)$.
Otherwise, assume $\delta_0(t)=|A_{ij}|$, $i\in S_a$,
and $j\in S_b$. We proceed by two cases.

If $\min\{|A_{i\ell}|, |A_{j\ell}|\}\leq \frac{\delta_0(t)}{2\alpha}$, then $\delta_0(t+1)\ge|A_{ij}'|\ge \frac{|A_{ij}|}{2(1+\alpha)}=\frac{\delta_0(t)}{2(1+\alpha)}$.
Otherwise, since agent $\ell$ cannot be both in the same cluster as agent $i$ and in the same cluster as agent $j$, we have
$\delta_0(t+1)\ge \min \{|A'_{j\ell}|, |A'_{i\ell}|\} \geq \min \{|A_{j\ell}|, |A_{i\ell}|\} \geq  \frac{\delta_0(t)}{2\alpha}\ge\frac{\delta_0(t)}{2(1+\alpha)}$.
In that step we used that
$|A'_{i\ell}|\ge |A_{i\ell}|$,
which follows since $\ell$ influenced
$i$: It can be checked directly 
from~\eqref{eq:new-aij}, but it also
follows since the new opinion lies on the
arc between $\vec u_i$ and $\vec u_j$
(if $A_{ij}>0$) or between
$\vec u_i$ and $-\vec u_j$ (otherwise).
Furthermore, we used
$|A'_{j\ell}|\ge |A_{j\ell}|$.
This holds since either $j=i$
and the previous argument applies,
or $j\ne i$, in which case
$A'_{j\ell}=A_{j\ell}$.
\end{proof}

\begin{lemma}
\label{lem:q1-lower-bound}
There exists $\eps>0$ such that
if $\cU^t$ is $(\eps,\eps)$-inactive
with $Q_0(t)<\infty$, 
then either it holds $\delta_1(t+1)\le O(\delta_0(t))$
or $\delta_1(t+1)\le O(\delta_1(t))$. In particular,
$Q_1(t+1)\ge\min(Q_1(t),Q_0(t))-O(1)$.
\end{lemma}

\begin{proof}
Let $\cU=\cU^t$, $\cU'=\cU^{t+1}$, and assume that
agent $\ell$ influenced agent $i\in S_b$.
By definition, $|A_{ij}|\ge 1-\delta^2_1(t)$
for every $j\in S_b$.

If $\ell\in S_b$, then for any agent $j\in S_b$, 
since $\sign(A_{ij})=
\sign(A_{i\ell})\sign(A_{j\ell})$
by \Cref{cl:consistent-signs},
we have
\begin{align}
|A'_{ij}|= \frac{|A_{ij}|+\alpha|A_{i\ell}A_{j\ell}|}{\sqrt{1+(2\alpha+\alpha^2)A^2_{i\ell}}}
\ge \frac{1-\delta^2_1(t)+\alpha(1-2\delta^2_1(t))}{1+\alpha}
= 1-\frac{1+2\alpha}{1+\alpha}\delta^2_1(t)\;.
\end{align}
Hence, $\delta_1(t+1)\le O(\delta_1(t))$. On the other hand,
if $\ell\notin S_b$, then, for any $j\in S_b$, it holds
\begin{align}
|A'_{ij}|\ge\frac{1-\delta^2_1(t)-\alpha\delta^2_0(t)}
{\sqrt{1+(2\alpha+\alpha^2)\delta_0^2(t)}}
\ge 1-O\left(\delta^2_1(t)+\delta^2_0(t)\right)\;,
\end{align}
which implies
$\delta_1(t+1)\le O\left(\sqrt{\delta^2_1(t)+\delta^2_0(t)}\right)
\le O(\max(\delta_0(t), \delta_1(t)))$.
\end{proof}

\begin{lemma}
\label{lem:q1-bounded}
There exist $\eps>0$ and $C>0$, such that if 
$\cU^t$ is $(\eps,\eps)$-inactive with
$\delta_1(t)\ge C\delta_0(t)>0$, then
$\Omega(\delta_1(t))\le \delta_1(t+1)\le O(\delta_1(t))$.
In particular, if $Q_1(t)\le Q_0(t)-\log C$, then
$|Q_1(t+1)-Q_1(t)|\le O(1)$.
\end{lemma}

\begin{proof}
For any fixed $C\ge 1$,
by \Cref{lem:q1-lower-bound}, if $\delta_1(t)\ge C\delta_0(t)$,
then it holds $\delta_1(t+1)\le O(\max(\delta_0(t),\delta_1(t)))\le O(\delta_1(t))$.

For the lower bound, let $\cU=\cU^t$ and $\cU'=\cU^{t+1}$
and assume that agent $\ell$ influenced $i$.
If $|A_{i'j}|=1-\delta^2_1(t)$
such that
$i\notin\{i',j\}$, then $\delta_1(t+1)\ge\delta_1(t)$.

On the other hand,
if $|A_{ij}|=1-\delta^2_1(t)$
for some $i,j\in S_b$,
again we proceed by cases.
If $\ell\notin S_b$ then 
$1-\delta^2_1(t+1)\le |A'_{ij}|\le |A_{ij}|+\alpha\delta^2_0(t)
=1-\delta^2_1(t)+\alpha\delta_0^2(t)\le 1-\delta^2_1(t)/2$,
where the last inequality holds if $\delta^2_1(t)\ge 2\alpha\delta^2_0(t)$, which is true if
$C\ge \sqrt{2\alpha}$. That implies $\delta_1(t+1)\ge\delta_1(t)/\sqrt{2}$.
If $\ell\in S_b$ and $|A_{i\ell}|>1-\frac{1+\alpha}{4(2\alpha+\alpha^2)}\delta^2_1(t)$, then
$A^2_{i\ell}>1-\frac{1+\alpha}{2(2\alpha+\alpha^2)}\delta^2_1(t)$, 
and
\begin{align}
|A'_{ij}|\le\frac{1-\delta_1^2(t)+\alpha}
{\sqrt{1+(2\alpha+\alpha^2)\left(
1-\frac{1+\alpha}{2(2\alpha+\alpha^2)}\delta^2_1(t)
\right)}}
=\frac{1-\frac{\delta_1^2(t)}{1+\alpha}}
{\sqrt{1-\frac{1}{2(1+\alpha)}\delta^2_1(t)}}
\le\frac{1-\frac{\delta_1^2(t)}{1+\alpha}}
{1-\frac{1}{2(1+\alpha)}\delta^2_1(t)}
\le 1-\frac{\delta^2_1(t)}{4(1+\alpha)}\;,
\end{align}
which implies $\delta_1(t+1)\ge\frac{\delta_1(t)}{2\sqrt{1+\alpha}}$.
If $\ell\in S_b$ and $|A_{i\ell}|\le1-\frac{1+\alpha}{4(2\alpha+\alpha^2)}\delta^2_1(t)$,
then, letting $|A_{i\ell}|=1-\delta^2$
\begin{align}
|A'_{i\ell}|=\frac{1-\delta^2+\alpha(1-\delta^2)}
{\sqrt{1+(2\alpha+\alpha^2)(1-\delta^2)^2}}
&\le\frac{1-\delta^2}{\sqrt{1-\frac{2\alpha+\alpha^2}{(1+\alpha)^2}2\delta^2}}
\le 1-\delta^2+\frac{0.5+2\alpha+\alpha^2}{(1+\alpha)^2}\delta^2\\
&\le 1-\Omega(\delta^2)\le 1-\Omega(\delta_1^2(t))\;,
\end{align}
which again gives $\delta_1(t+1)\ge \Omega(\delta_1(t))$.
\end{proof}

We now turn to proving the two properties
of expectation, namely
\eqref{eq:p0-martingale} 
and~\eqref{eq:p1-martingale}.
As these properties are stated for a larger
number of steps $T$,
the following corollary
will be useful to control $Q_0$ and $Q_1$ over several time steps:

\begin{corollary}
\label{cor:bounded-q0-q1}
There exists $\Cstep\ge 1$ such that for all $\eps'>0$ and $T$,
there exists $\eps(\eps', T)$ with the following property.
If $\cU^{t}$ is $(\eps,\eps)$-inactive and $Q_0(t)<\infty$, then for every time step
$t\le t'\le t+T$, configuration $\cU^{t'}$ remains
$(\eps',\eps')$-inactive and furthermore:
\begin{enumerate}
    \item $\delta_0(t')/\Cstep\le\delta_0(t'+1)\le
    \Cstep\delta_0(t')$.
    \item $\delta_1(t'+1)\le\Cstep\max(\delta_0(t'),\delta_1(t'))$.
    \item If $\delta_1(t')\ge\Cstep\delta_0(t')$, then
    $\delta_1(t')/\Cstep\le\delta_1(t'+1)\le
    \Cstep\delta_1(t')$.
\end{enumerate}
\end{corollary}
\begin{proof}
From \Cref{lem:q0-bounded},
\Cref{lem:q1-lower-bound}, and \Cref{lem:q1-bounded},
there exist $\eps''>0$ and $\Cstep\ge 1$ such that
the properties 1--3 simultaneously hold whenever the configuration
$\cU^{t'}$ is $(\eps'',\eps'')$-inactive.
Let us take $\eps=\min(\eps',\eps'')/\Cstep^{T}$.

Assume that $\cU^t$ is $(\eps,\eps)$-inactive.
By induction, applying \Cref{lem:q0-bounded}
and \Cref{lem:q1-lower-bound},
it follows that $\cU^{t'}$ is
$(\Cstep^{t'-t}\eps,\Cstep^{t'-t}\eps)$-inactive
for every $t\le t'\le t+T$.
In particular,
$\cU^{t'}$ is both
$(\eps',\eps')$-inactive and $(\eps'',\eps'')$-inactive,
which implies that it satisfies properties 1-3.
\end{proof}

We turn to proving~\eqref{eq:p1-martingale}.
In the proof, we will apply a result proved
in~\cite{ABH24}. This result reflects the fact
that inside-cluster interactions can only increase
absolute correlations between opinions in a cluster.

\begin{lemma}[Claim~3.12 and Claim~3.13 in~\cite{ABH24}]
\label{lem:strictly-convex-increasing}
Let $n\ge 2$ and $\cU$ be a configuration that
satisfies $|A_{ij}|>\sqrt{2}/2$ for every $i,j\in[n]$.
Let $\cU'$ be obtained from $\cU$ by agent $\ell$ influencing $i$.
Then, for every $j$, it holds
$|A'_{ij}|\ge\min(|A_{ij}|,|A_{j \ell}|)$.

Furthermore, there exists a sequence of $K=\binom{n}{2}$ 
interactions and a constant $c=c(\alpha)<1$ such that
\begin{align}
\label{eq:13}
\max_{1\le i,j\le n}(1-|A_{ij}^{K}|)
\le c\max_{1\le i,j\le n}(1-|A_{ij}|)\;.
\end{align}
\end{lemma}

\begin{lemma}
\label{lem:q1-expectation}
There exists $T_0$ such that, for every
$T\ge T_0$, there exist positive constants $\eps=\eps(T)$, $C=C(T)$, $c=c(T)$
with the following property:

If $\cU^t$ is $(\eps,\eps)$-inactive and satisfies
$\delta_1(t)\ge C\delta_0(t)>0$, then
$\cU^{t'}$ remains $(\epsbase,\epsbase)$-inactive
for $t\le t'\le t+T$
and
$\expec[Q_1(t+T)\;\vert\;\cU^t]\ge Q_1(t)+c$.
\end{lemma}

\begin{proof}
While the details require some care,
the idea of the proof is simpler. First, 
we show that there exists a sequence of at most $\binom{n}{2}$
interactions which decrease $\delta_1$ by
a constant factor. On the other hand,
we will see that it holds
$\delta_1^2(t+1)\le \delta_1^2(t)+O(\delta_0^2(t))$. Choosing
$C$ sufficiently large, the $\delta_0^2$
terms become sufficiently small to conclude
that, over $T$ steps, 
$\delta_1$ can grow only by an
arbitrarily small amount.
Since, as mentioned, $\delta_1$ decreases by a noticeable
amount with noticeable probability, the bound
on the expectation follows.

Let $T_0=\binom{n}{2}$ and $T\ge T_0$.
Take $\eps=\eps(\epsbase,T)$ from
\Cref{cor:bounded-q0-q1}.
Consider an $(\eps,\eps)$-inactive configuration $\cU^t$.
By \Cref{cor:bounded-q0-q1}, configuration
$\cU^{t'}$ remains $(\epsbase,\epsbase)$-inactive for
$t\le t'\le t+T$.

By the second part of \Cref{lem:strictly-convex-increasing},
for each cluster $S_a$, there exists
a sequence of $\binom{|S_a|}{2}$ interactions inside that cluster after 
which~\eqref{eq:13} is satisfied.
It follows that after the total sequence
of $\hat{T}=\sum_{a=1}^k \binom{|S_a|}{2}$
interactions
the configuration
$\cU^{t+\hat{T}}$ satisfies 
\begin{align}
\max_{\substack{i,j\in S_a\\ 1\le a\le k}}
1-\left|A^{t+\hat{T}}_{ij}\right|
&\le c_\alpha^2
\max_{\substack{i,j\in S_a\\ 1\le a\le k}}
1-\left|A_{ij}^t\right|
\end{align}
for some $c_\alpha<1$. From this it
follows $\delta_1(t+\hat{T})\le c_\alpha\delta_1(t)$
and $Q_1(t+\hat{T})\ge Q_1(t)-\log c_\alpha$.
Furthermore, as $\sum_{a=1}^k\binom{|S_a|}{2}\le\binom{n}{2}=T_0$
and as by the first part of \Cref{lem:strictly-convex-increasing}
applying additional interactions inside clusters
does not increase $\delta_1$, 
there also exists a sequence
of $T_0$ interactions satisfying
$Q_1(t+T_0)\ge Q_1(t)-\log c_\alpha$.
Since such a sequence occurs with probability
$p=n^{-2T_0}$, it follows
\begin{align}
\label{eq:79}
\Pr\left[Q_1(t+T_0)\ge Q_1(t)-\log c_\alpha
\;\big\vert\;\cU^t\right]\ge p\;.
\end{align}

On the other hand, for an
$(\epsbase,\epsbase)$-inactive configuration $\cU$ at some time,
let agent $\ell\in S_a$ influence agent $i\in S_b$
and call the new configuration $\cU'$.
If $a=b$, then it follows from \Cref{lem:strictly-convex-increasing} that $\delta_1(\cU')\le\delta_1(\cU)$.
If $a\ne b$, then for every $j\in S_b$, 
by~\eqref{eq:new-aij} we have
$|A'_{ij}|\ge
\frac{|A_{ij}|-\alpha\delta_0(t)^2}
{\sqrt{1+(2\alpha+\alpha^2)\delta_0(t)^2}}$,
which implies
$|A'_{ij}|\ge |A_{ij}| -(3\alpha+\alpha^2)\delta^2_0(\cU)$.
Putting the two cases together, there exists some $C_\alpha>0$ such
that $|A'_{ij}|\ge |A_{ij}|-C_\alpha \delta^2_0(\cU)$,
consequently 
\begin{align}
\label{eq:14}
\delta^2_1(\cU')\le \delta^2_1(\cU)+C_\alpha \delta^2_0(\cU)\;.
\end{align}
Now, given $T\ge T_0$,
we set $C=\Cstep^{2T_0}\cdot \sqrt{\frac{2TC_\alpha\Cstep^{2T}}{p\log c_\alpha^{-1}}}$. Assume that $\cU^t$
is $(\eps,\eps)$-inactive and satisfies
$\delta_1(t)\ge C\delta_0(t)$. Clearly,
that implies
\begin{align}
\label{eq:15}
\delta_0^2(t)\le
C^{-2}\cdot \delta_1^2(t)\le \frac{p\log c_\alpha^{-1}}{2TC_\alpha\Cstep^{2T}}\cdot \delta_1^2(t)
\;.
\end{align}
Furthermore, from
\Cref{cor:bounded-q0-q1}, both
$\delta_0(t)$ and $\delta_1(t)$
can change by at most factor $\Cstep$
in one step. Hence, $\Cstep^{2T_0}\delta_1(t+T_0)\ge C\delta_0(t+T_0)$,
and
\begin{align}
\label{eq:16}
\delta_0^2(t+T_0)\le
\frac{p\log c_\alpha^{-1}}{2TC_\alpha\Cstep^{2T}}\cdot \delta_1^2(t+T_0)
\;.
\end{align}
Applying~\eqref{eq:14},
\Cref{cor:bounded-q0-q1},
and~\eqref{eq:15},
\begin{align}
\label{eq:80}
\delta^2_1(t+T)\le \delta^2_1(t)+
C_\alpha\sum_{t'=0}^{T-1}\delta^2_0(t+t')
\le \delta^2_1(t)+
TC_\alpha\Cstep^{2T}\delta^2_0(t)
\le \delta_1^2(t)\left(1+\frac{p}{2}\log c_\alpha^{-1}\right)\;,
\end{align}
hence
\begin{align}
\label{eq:83}
Q_1(t+T)&\ge Q_1(t)-\frac{1}{2}\log\left(
1+\frac{p}{2}\log c_\alpha^{-1}
\right)\ge Q_1(t)-\frac{p}{4}\log c_\alpha^{-1}\;.
\end{align}
On the other hand, due to~\eqref{eq:79}, with probability at least
$p$ it holds $Q_1(t+T_0)\ge Q_1(t)+\log c_\alpha^{-1}$.
Redoing the calculation in~\eqref{eq:80},
but replacing~\eqref{eq:15} 
with~\eqref{eq:16}, it follows
\begin{align}
\delta_1^2(t+T)\le 
\delta_1^2(t+T_0)+C_\alpha\sum_{t'=0}^{T-T_0-1}\delta_0^2(t+T_0+t')\le
\delta_1^2(t+T_0)
\left(1+\frac{p}{2}\log c_\alpha^{-1}\right)\;,
\end{align}
which implies $Q_1(t+T)\ge Q_1(t+T_0)-\frac{p}{4}\log c_{\alpha}^{-1}$.
Hence, with probability at least $p$
it holds
\begin{align}
\label{eq:81}
Q_1(t+T)&\ge Q_1(t)+\log c_\alpha^{-1}-\frac{p}{4}\log c_{\alpha}^{-1}
\ge Q_1(t)+\frac{3}{4}\log c_\alpha^{-1}\;.
\end{align}
Putting~\eqref{eq:83} and~\eqref{eq:81} together, we conclude
\begin{align}
\expec[Q_1(t+T)\;\vert\;\cU^t]
&\ge Q_1(t)+\frac{3p}{4}\log c_\alpha^{-1}
-\frac{p}{4}\log c_\alpha^{-1}
\ge Q_1(t)+\frac{p}{2}\log c_\alpha^{-1}\;.\qedhere
\end{align}
\end{proof}

\section{Consistent configurations and expectation of $Q_0$}
\label{sec:q0-expectation}

We turn to establishing the expectation inequality
$\expec [Q_0(t+T)\;\vert\;\cU^t]\ge Q_0(t)+c$.
We will proceed according to the outline
explained in \Cref{sec:plan}.
Accordingly, we will use a concept of a consistent configuration. In fact, we need a slightly
more general definition compared to
the one given by~\eqref{eq:23}.

\begin{definition}
Let $\cU$ be an $(\epsbase,\epsbase)$-inactive configuration, $a\ne b$ be indices
of two clusters of $\cU$, and $m\ge 0$. We say that
$\cU$ is \emph{$(a,b,m)$-consistent} 
if
for all $i,i'\in S_a,j,j'\in S_b$ it holds:
\begin{enumerate}
\item 
$\sign A_{i'j'}=\sign A_{ii'}\sign A_{ij}\sign A_{jj'}\ne 0$.
\item $|A_{i'j'}|\ge m\delta_0(\cU)$. 
\end{enumerate}

We also say that $\cU$ is $(a,b)$-consistent if it is $(a,b,0)$-consistent.
If $\cU$ is $(a,b)$-consistent and
$\cU'$ is reachable in one step from $\cU$, we say that
$\cU'$ \emph{remains
consistent} if it is also $(a,b)$-consistent with
$\sign A'_{ij}=\sign A_{ij}$ for every $i\in S_a,j\in S_b$.
For $\cU$ an $(a,b)$-consistent configuration we define
\begin{align}
\label{eq:x01}
\deltars(\cU):=\min_{i\in S_a,j\in S_b}|A_{ij}|\;.
\end{align}
\end{definition}

\begin{claim}
\label{cl:new-aij}
Let $\cU$ be a configuration with $|A_{i\ell}|\ge 1/2$ 
and $\cU'$ be reachable in one step from $\cU$ by $\ell$
influencing $i$. Then, for every agent $j$:
\begin{enumerate}
\item If $A_{ij}<0$ and $A_{i\ell}A_{j\ell}\ge 0$, 
then $A'_{ij}\ge A_{ij}+\frac{\alpha}{2(1+\alpha)}|A_{j\ell}|$.
\item If $A_{ij}\ge 0$ and $A_{i\ell}A_{j\ell}\ge 0$, then $A'_{ij}\ge\frac{\alpha}{2(1+\alpha)}|A_{j\ell}|$.
\end{enumerate}
\end{claim}
\begin{proof}
    \begin{enumerate}
        \item
        Using $A_{ij}<0$, it holds $A_{ij}'
        = \frac{A_{ij} + \alpha A_{i\ell} A_{j\ell}}{\sqrt{1+(2\alpha+\alpha^2)A_{i\ell}^2}}
        \ge A_{ij}+\frac{\alpha}{2(1+\alpha)}|A_{j\ell}|$.
        \item Similarly, but this time using $A_{ij}\ge 0$,
        it holds
        $A'_{ij}\ge \frac{\alpha}{2(1+\alpha)}|A_{j\ell}|$.
        \qedhere
    \end{enumerate}
\end{proof}

\begin{lemma}
\label{lem:consistent-reachable}
There exist $\eps>0$, $\cact>0$, and $K$ such that the following holds:
Let $\cU^t$ be $(\eps,\eps)$-inactive with
$Q_0(t)<\infty$, in particular $\cU^t$
has at
least two clusters.
Let $S_a$,
$S_b$ be the clusters realizing 
$\delta_0(t)=\max_{i\in S_a,j\in S_b}|A_{ij}^t|$.
Then, there
exists a sequence of $K$ interactions such that
$\cU^{t+K}$ is $(a,b,\cact)$-consistent.
\end{lemma}

\begin{proof}
Let $K_0=\lceil\frac{2(1+\alpha)}{\alpha}\rceil+1$,
$K_1=\lceil \frac{4(1+\alpha)^2\Cstep^{K_0}}{\alpha^2}\rceil+1$
and $K=n\cdot (K_0+K_1)$.
Then, let $\eps=\eps(\epsbase, K)$ from
\Cref{cor:bounded-q0-q1}.

Let $\cU=\cU^t$ and let $i_0\in S_a,j_0\in S_b$ such that
$|A_{i_0j_0}|=\delta_0(t)$. 
We propose the following sequence of interactions:
First, let $i_0$ influence every agent $i\in S_a$
for $K_0$ times.
Let us call the new intermediate configuration $\Utilde$.
Then, for every $j\in S_b$, let $j_0$ influence $j$
at least $K_1$ times, such that the total
number of interactions is $K$. 
Let us call the final configuration $\Uhat$.

Due to symmetry we can assume w.l.o.g.~that $A_{i_0j_0}>0$, and $A_{ii_0}>0,A_{jj_0}>0$
for every $i\in S_a,j\in S_b$. Accordingly, to show that $\Uhat$ is 
$(a,b,\cact)$-consistent
it is sufficient to prove
$\Ahat_{ij}>0$ and 
$\Ahat_{ij}\ge \cact \delta_0(\Uhat)$ for every
$i,i'\in S_a,j,j'\in S_b$.

First, let $i\in S_a$ and let us analyze
$\Atilde_{ij_0}$ in the intermediate configuration
$\Utilde$.
Applying \Cref{cl:new-aij} for $\ell=i_0$ and $j=j_0$,
and observing that by assumption $|A_{ij_0}|\le A_{i_0j_0}$,
it follows that
$\Atilde_{ij_0}\ge\frac{\alpha}{2(1+\alpha)}A_{i_0j_0}$.
By \Cref{cor:bounded-q0-q1}, 
it also holds
$\Atilde_{ij_0}\le \Cstep^{K_0}A_{i_0j_0}$.

Let us move on to the configuration $\Uhat$.
Let $i\in S_a$ and $j\in S_b$.
By \Cref{cor:bounded-q0-q1} and the preceding calculation,
\begin{align}
|\Atilde_{ij}|
\le \Cstep^{K_0}A_{i_0j_0}\le\Cstep^{K_0}\frac{2(1+\alpha)}
{\alpha}\Atilde_{ij_0}
\end{align}
Applying \Cref{cl:new-aij} for $\ell=j_0$, $i=j$ and $j=i$,
if $\Atilde_{ij}<0$, then
after (at least) $K_1-1$ interactions of $j_0$ influencing $j$,
it holds 
\begin{align}
\Ahat_{ij}\ge\min\left(0, \Atilde_{ij}+(K_1-1)\frac{\alpha}{2(1+\alpha)}\Atilde_{ij_0}\right)
\ge \min\left(0, \Atilde_{ij}+(K_1-1)\frac{\alpha^2}{4(1+\alpha)^2\Cstep^{K_0}}|\Atilde_{ij}|\right)
\ge 0\;.
\end{align}
Therefore, regardless of the sign of $\Atilde_{ij}$,
after $K_1$ interactions it holds
$\Ahat_{ij}\ge \frac{\alpha}{2(1+\alpha)}\Atilde_{ij_0}
\ge\frac{\alpha^2}{4(1+\alpha)^2}A_{i_0j_0}
\ge\frac{\alpha^2}{4(1+\alpha)^2\Cstep^K}\delta_0(t+K)$,
where the last inequality follows
by a crude application of \Cref{cor:bounded-q0-q1}.
Indeed, that implies that $\cU^{t+K}$ is
$(a,b,\cact)$-consistent for
$\cact=\frac{\alpha^2}{4(1+\alpha)^2\Cstep^{K}}$.
\end{proof}

\begin{claim}
\label{cl:aij-consistent}
There exists $\eps'>0$ such that:
Let $\cU$ be an
$(\eps',\eps')$-inactive and $(a,b)$-consistent configuration, $i\in S_a,j\in S_b$
and $\cU'$ a configuration obtained in one step from $\cU$ by $\ell$ influencing $i$.
\begin{enumerate}
    \item If $\ell\notin S_a\cup S_b$, then
    $|A'_{ij}-A_{ij}|\le (3\alpha+\alpha^2)\delta_0^2(\cU)$.
    \item If $\ell\in S_b$, then $\sign A'_{ij}=\sign A_{ij}$
    and $|A'_{ij}|\ge|A_{ij}|$.
    \item If $\ell\in S_a$, then $\sign A'_{ij}=\sign A_{ij}$
    and $|A'_{ij}|\ge\min(|A_{ij}|,|A_{j\ell}|)$.
\end{enumerate}
\end{claim}

\begin{proof}\leavevmode
\begin{enumerate}
    \item 
    Let $\delta=\delta_0(\cU)$.
    Assume that $A_{ij}\ge 0$.
    By~\eqref{eq:new-aij} and using
$\frac{1}{\sqrt{1+x}}\ge 1-x$, it holds
$A'_{ij}\ge \frac{A_{ij}}{\sqrt{1+(2\alpha+\alpha^2)\delta^2}}-\alpha\delta^2
\ge A_{ij}-(3\alpha+\alpha^2)\delta^2$.
Similarly,
$A'_{ij}\le A_{ij}+\alpha\delta^2$.
Therefore,
$|A'_{ij}-A_{ij}|\le
(3\alpha+\alpha^2)\delta^2$.
A similar calculation obtains for $A_{ij}<0$.
\item
First $\sign A'_{ij}=\sign A_{ij}$ follows 
from~\eqref{eq:new-aij}
as
$\sign A_{ij}=\sign A_{i\ell}\cdot \sign A_{j\ell}$ by consistency.
Furthermore, we have
\begin{align}
|A'_{ij}|\ge\left(|A_{ij}|+\frac{\alpha}{2}|A_{i\ell}|\right)(1-(2\alpha+\alpha^2)A_{i\ell}^2)
\ge|A_{ij}|\;,
\end{align}
where in the last step we used that $|A_{i\ell}|\le \eps'$
for a sufficiently small fixed $\eps'$.
\item
Again, $\sign A'_{ij}=\sign A_{ij}$
follows by consistency from $\sign A_{ij}=\sign A_{i\ell}\cdot\sign A_{j\ell}$,
and then we have
\begin{align}
|A_{ij}|\ge\min(|A_{ij}|,|A_{j\ell}|)\cdot
\frac{1+\alpha|A_{i\ell}|}{\sqrt{1+(2\alpha+\alpha^2)A_{i\ell}^2}}
\ge\min(|A_{ij}|,|A_{j\ell}|)\;,
\end{align}
where the last inequality holds since
$(1+\alpha x)^2=1+2\alpha x+\alpha^2 x^2
\ge 1+(2\alpha+\alpha^2)x^2$ for $0\le x\le 1$.\qedhere
\end{enumerate}
\end{proof}

Recall that for an $(a,b)$-consistent
configuration, we defined
$\deltars(\cU)=\min_{i\in S_a,j\in S_b}|A_{ij}|$
In the next two lemmas we study this quantity.
First, we show that there exists 
a fixed length sequence of interactions
that increases $\deltars$ noticeably. Then,
we show that over any constant number of interactions,
the configuration remains consistent and furthermore
$\deltars$ cannot decrease by more than
a negligible amount.

\begin{lemma}
\label{lem:increase-delta-min}
There exist $\eps>0$, $\cadv>0$, and 
$K$ such that the following
holds. Let $\cU^t$ be an $(\eps,\eps)$-inactive
configuration that remains
$(a,b)$-consistent for any sequence of $K$ interactions.
Then, there exists a sequence of $K$ interactions
such that
$\deltars(t+K)\ge (1+\cadv)\cdot\deltars(t)$.
\end{lemma}

\begin{proof}
Let $K=n^2$ and $\eps=\eps(\eps',K)$, where
$\eps'$ comes from \Cref{cl:aij-consistent}
and $\eps(\eps',K)$ from \Cref{cor:bounded-q0-q1}.
In particular, the configuration remains
$(\eps',\eps')$-inactive for $t\le t'\le t+K$.

Let us take any sequence of $K$ interactions 
where all interactions are between $S_a$
and $S_b$ and, furthermore,
for every $i\in S_a$ and $j\in S_b$,
agent $i$ influences $j$ at least once (and perhaps multiple times
so that the total number of interactions is $K$).

By \Cref{cl:aij-consistent}, for every $i\in S_a,j\in S_b$,
the sequence $|A^{t'}_{ij}|$ is nondecreasing for $t\le t'\le t+K$.
Furthermore,
there exists at least one time $t'$ where,
applying~\eqref{eq:new-aij},
\begin{align}
\left|A^{t'+1}_{ij}\right|
&\ge\frac{(1+\alpha)|A^{t'}_{ij}|}{
\sqrt{1+(2\alpha+\alpha^2)(A^{t'}_{ij})^2}}
\ge \frac{1+\alpha}{\sqrt{1+(2\alpha+\alpha^2)(\eps')^2}}\left|A_{ij}^{t'}\right|\;.
\end{align}
Accordingly, it holds $|A^{t+K}_{ij}|\ge (1+\cadv)|A^t_{ij}|$
and $\deltars(t+K)\ge (1+\cadv)\cdot\deltars(t)$
for $\cadv=\frac{1+\alpha}{\sqrt{1+(2\alpha+\alpha^2)(\eps')^2}}-1$.
\end{proof}

\begin{lemma}
\label{lem:remain-consistent}
Let $\cact>0$ be the constant from \Cref{lem:consistent-reachable}.
For every $0<c<1$ and $T$, there exists $\eps=\eps(c,T)>0$ such that
if $\cU^{t}$ is $(\eps,\eps)$-inactive and
$(a,b,\cact)$-consistent,
then the configuration
$\cU^{t+t'}$ remains $(a,b)$-consistent
for $t\le t'\le t+T$.
Furthermore, for every $t\le t'\le t''\le t+T$, it holds
$\deltars(t'')\ge (1-c)\cdot\deltars(t')$.
\end{lemma}

\begin{proof}
Let $\eps'$ come from \Cref{cl:aij-consistent}
and take $\eps''=\eps(\eps',T)$ from \Cref{cor:bounded-q0-q1}.
Then, let us take
\begin{align}
\eps=\min\left(\eps'',\frac{c\cdot \cact}{2T(3\alpha+\alpha^2)\Cstep^{2T}}\right)\;.
\end{align}

Assume that the configuration is $(a,b)$-consistent
at time $t'$ and that agent $\ell$
influences agent $i_0$ at that time.
If $i_0\notin S_a\cup S_b$, then no relevant correlations change and 
$A_{ij}^{t'+1}=A_{ij}^{t'}$ for every
$i\in S_a,j\in S_b$.
On the other hand, assume that $i_0\in S_a\cup S_b$.
By \Cref{cl:aij-consistent}, if $\ell\notin S_a\cup S_b$, then for every $i\in S_a$ and $j\in S_b$ it holds
$|A^{t'+1}_{ij}-A^{t'}_{ij}|\le 
(3\alpha+\alpha^2)\delta_0^2(t')$.
If $\ell\in S_a\cup S_b$, then by
\Cref{cl:aij-consistent}, it follows
for every $i\in S_a,j\in S_b$ that
$\sign(A_{ij}^{t'+1})=\sign(A_{ij}^{t'})$ and
$|A^{t'+1}_{ij}|\ge \min(|A_{ij}^{t'}|,
|A_{j\ell}^{t'}|,|A_{i\ell}^{t'}|)$.

Assume that $\cU^t$ is $(\eps,\eps,\cact)$-consistent
at time $t$.
By symmetry, let us assume w.l.o.g.~that
$A_{ij}^t>0$ for every $i\in S_a$, $j\in S_b$.
Let $t\le t'\le t+T$.
By applying the reasoning above,
as well as \Cref{cor:bounded-q0-q1}
inductively,
it holds
\begin{align}
\label{eq:84}
\min_{i\in S_a,j\in S_b}A_{ij}^{t'}\ge 
\min_{i\in S_a,j\in S_b}A_{ij}^t-\sum_{s=t}^{t'-1}(3\alpha+\alpha^2)
\delta_0^2(s)
\ge \cact \delta_0(t)-T(3\alpha+\alpha^2)\Cstep^{2T}
\delta_0^2(t)
\ge(1-c/2)\cact\delta_0(t)\;.
\end{align}
In particular $\sign A_{ij}^t=\sign A_{ij}^{t'}$
and the configuration remains consistent.

Similarly, let $t\le t'\le t''\le t+T$, $i\in S_a$ and
$j\in S_b$. 
From~\eqref{eq:84}, note that
$\delta_0(t)\le \frac{\deltars(t')}{\cact\cdot (1-c/2)}
\le \deltars(t')\cdot\frac{1+c}{\cact}$. Hence,
\begin{align}
\deltars(t'')&\ge \deltars(t')-T(3\alpha+\alpha^2)\Cstep^{2T}
\delta_0^2(t)
\ge\deltars(t')-\frac{c\cdot\cact}{2}\delta_0(t)\\
&\ge\deltars(t')-\frac{c(1+c)}{2}\deltars(t')
\ge (1-c)\deltars(t')\;,
\end{align}
which concludes the proof.
\end{proof}

\begin{lemma}
\label{lem:q0-expectation}
There exists $\eps>0$ and $T$ such that,
if $\cU^t$ is $(\eps,\eps)$-inactive and $Q_0(t)<\infty$, then
$\cU^{t'}$ remains $(\epsbase,\epsbase)$-inactive for $t\le t'\le t+T$ and
$\expec[Q_0(t+T)\;\vert\;\cU^t]\le Q_0(t)-\Omega(1)$.
\end{lemma}

\begin{proof}
Take $K$ which is the maximum\footnote{
Note that the sequence of $K$ interactions 
that exists by \Cref{lem:consistent-reachable}
can be extended to a longer sequence by
adding interactions of the form
$(i,i)$ that do not change the configurations.
The same goes for \Cref{lem:increase-delta-min}.
Therefore, the relevant sequences both exist
and have the claimed properties for 
$K$ taken to be the maximum.
}
of $K$
from \Cref{lem:consistent-reachable}
and \Cref{lem:increase-delta-min}. 
Then, take $T=M\cdot K$ for sufficiently large $M=M(d,\alpha,n)$
(as will be seen below). Then, take $\eps'>0$ to be the minimum of epsilons for which \Cref{lem:consistent-reachable} and
\Cref{lem:increase-delta-min} hold. 
Recall the constants $\cact$ and $\cadv$
from those lemmas.

Let $p=n^{-2K}$, $c=1-(1+\cadv)^{-p/2}$ and
let $\eps''=\eps(c,T)$ from \Cref{lem:remain-consistent}.
Finally, take $\eps=\eps(\min(\eps'\eps''),T)$ from \Cref{cor:bounded-q0-q1}. In particular, if
at any time $\cU^{t'}$ becomes $(a,b,\cact)$-consistent,
then it remains $(a,b)$-consistent until time $t+T$.

Let
\begin{align}
W&=\min\left\{1\le m\le M: \cU^{t+mK}\text{ is $(a,b,\cact)$-consistent for some $a,b$}
\right\}\;,
\end{align}
and $W=M$ if the configuration does not become $(a,b,\cact)$-consistent
for any of $1\le m\le M$.
By \Cref{lem:consistent-reachable}, at every time
step $t'$, it holds
$\Pr[\cU^{t'+K}\text{is $(a,b,\cact)$-consistent}\;\vert\;\cU^{t'}]\ge p$.
That implies, conditioned on $\cU^t$,
\begin{align}
\expec W
&=\sum_{m=1}^M \Pr[W\ge m]
\le\sum_{m=1}^{\infty}(1-p)^{m-1}=\frac{1}{p}\;.
\end{align}
Furthermore, by \Cref{lem:remain-consistent}, if $\cU^{t+mK}$ is
$(a,b,\cact)$-consistent, then it remains consistent for all
$t+mK\le t'\le t+T$. Now, condition on some $\cU^{t+mK}$
for $W\le m<M$.
By \Cref{lem:increase-delta-min},
with probability at least $p$ it holds
$-\log\deltars(t+(m+1)K)\le-\log\deltars(t+mK)-\log(1+\cadv)$.
On the other hand, by \Cref{lem:remain-consistent}, it always
holds
\begin{align}
-\log\deltars(t+(m+1)K)\le-\log\deltars(t+mK)-\log(1-c)
=-\log\deltars(t+mK)+\frac{p}{2}\log(1+\cadv)\;.
\end{align}
Putting it together,
\begin{align}
\label{eq:85}
\expec\left[ -\log\deltars(t+(m+1)K)\;\vert\;\cU^{t+mK}\right]
\le -\log\deltars(t+mK)-\frac{p}{2}\log(1+\cadv)\;.
\end{align}
Therefore, applying \eqref{eq:85},
\Cref{cor:bounded-q0-q1}, and the fact
that $\deltars(t+WK)\ge\cact \delta_0(t+WK)$
(for $W<M$)
and $\deltars(t+T)\le\delta_0(t+T)$,
\begin{align}
\expec\left[Q_0(t+T)-Q_0(t)\;\vert\;\cU^t\right]
&\le
\sum_{m=1}^M\Pr[W=m]\cdot\left(
mK\log\Cstep+\expec\left[Q_0(t+T)-Q_0(t+mK)\;\vert\;W=m\right]
\right)\\
&\le \sum_{m=1}^M\Pr[W=m]\cdot\Big(
mK\log\Cstep-\log\cact\nonumber\\
&\qquad\qquad+\expec\left[-\log\deltars(t+T)+\log\deltars(t+mK)\;\vert\;W=m\right]
\Big)\\
&\le\sum_{m=1}^M
\Pr[W=m]\cdot\left(mK\log\Cstep-m\log\cact-(M-m)
\frac{p}{2}\log(1+\cadv)
\right)\\
&\le -M\frac{p}{2}\log(1+\cadv)
+(\expec W)\cdot\left(K\log\Cstep-\log\cact+
\frac{p}{2}\log(1+\cadv)
\right)\label{eq:82}\\
&\le -\frac{p}{4}\log(1+\cadv)\;,
\end{align}
where the last step follows after choosing
sufficiently large $M=M(d,\alpha,n)$,
as all other constants in~\eqref{eq:82} depend
only on $d$, $n$ and $\alpha$.
\end{proof}

\section{Taking $T$ steps at once
and martingale concentration}
\label{sec:many-steps-at-once}

Let us sum up what we proved so far.
The following statement follows from
\Cref{lem:q0-expectation},
\Cref{lem:q1-expectation}
and \Cref{cor:bounded-q0-q1}.

\begin{corollary}
\label{cor:q-final}
There exist $\eps>0$, $C'\ge 1$ and $T$ such that:
Let $\cU^t$ be an $(\eps,\eps)$-inactive configuration
with $Q_0(t)<\infty$.
Then, $\cU^{t'}$ remains $(\epsbase,\epsbase)$-inactive for $t\le t'\le t+T$.
Furthermore, it holds:
\begin{align}
|Q_0(t+1)-Q_0(t)|&\le C'\;,
\label{eq:q0-bounded-repeat}\\
\expec[Q_0(t+T)\;\vert\;\cU^{t}]&\le Q_0(t)-1/C'\;,
\label{eq:q0-expectation-repeat}\\
Q_1(t+1)&\ge \min(Q_0(t),Q_1(t))-C'\;.
\label{eq:q1-lower-bound-repeat}
\end{align}
Furthermore, if $Q_1(t)\le Q_0(t)-C'$, then:
\begin{align}
|Q_1(t+1)-Q_1(t)|&\le C'\;,
\label{eq:q1-bounded-repeat}\\
\expec[Q_1(t+T)\;\vert\;\cU^{t}]&\ge Q_1(t)+1/C'\;.
\label{eq:q1-expectation-repeat}
\end{align}
Finally, \eqref{eq:q0-bounded-repeat} and
\eqref{eq:q1-lower-bound-repeat} also hold
for every $t\le t'\le t+T$,
and \eqref{eq:q1-bounded-repeat} holds in the sense
that if $Q_1(t')\le Q_0(t')-C'$, then
$|Q_1(t'+1)-Q_1(t')|\le C'$.
\end{corollary}

\Cref{cor:q-final} allows to deduce
\Cref{lem:p0-p1-properties}:

\begin{proof}[Proof of \Cref{lem:p0-p1-properties}]
Take $\eps$ and $T$ from \Cref{cor:q-final}
and $C=2C'T$.
From~\eqref{eq:q0-bounded-repeat} it follows
$|P_0(t+1)-P_0(t)|\le C'T\le C$
and from \eqref{eq:q0-expectation-repeat} we have
$\expec[P_0(t+1)\;\vert\;\cU^{tT}]\le P_0(t)-1/C'\le P_0(t)-1/C$.
Applying~\eqref{eq:q0-bounded-repeat}
and~\eqref{eq:q1-bounded-repeat} inductively, we have
$Q_1(tT+k)\ge \min(Q_0(tT),Q_1(tT))-kC'$
for $k<T$,
hence it holds $P_1(t+1)\ge \min(P_0(t),P_1(t))-TC'$.

Finally, let $P_1(t)\le P_0(t)-2C'T$. 
Applying \eqref{eq:q0-bounded-repeat}
and \eqref{eq:q1-bounded-repeat} by induction, it holds
$Q_1(tT+k)\le Q_0(tT+k)-2C'T+2C'k<Q_0(tT+k)-C'$ for $k<T$.
Hence, by \eqref{eq:q1-bounded-repeat}, it holds
$|P_1(t+1)-P_1(t)|\le C'T$.
Finally, $\expec[P_1(t+1)\;\vert\;\cU^{tT}]\ge
P_1(t)+1/C'$ follows immediately from~\eqref{eq:q1-expectation-repeat}.
\end{proof}

As explained in \Cref{sec:plan},
in the rest of the proof we proceed more generally
and prove \Cref{thm:random-process}.
In that proof
we will need a simple consequence of the Azuma's inequality:

\begin{lemma}
\label{lem:modified-azuma}
For all $c_1,c_2>0$ there exist $c_3>0$ such that the following holds.

Let $X(t)$ be a random process adapted to a filtration $(\cF_t)_t$.
Assume that for all times $t$ almost surely
\begin{align}
|X(t)-X(t+1)|&\le c_1\;,\\
\label{eq:martingale_lemma_exp}
\expec\left[X(t+1)\;\vert\;\cF_t\right]&\le X(t)-c_2\;.
\end{align}
Then, for every integer $t\ge 0$ it holds
\begin{align}
\Pr[X(t)\ge X(0)-c_2t/2]\le\exp(-c_3t)\;.
\end{align}
\end{lemma}

\begin{proof}
We will apply the Azuma-Hoeffding inequality:
If a random process $Y(t)$ satisfies $|Y(t+1)-Y(t)|\le C$ and
$\expec[Y(t+1)\;\vert\;\cF_t]=Y(t)$ almost surely for every $t$,
then $\Pr[Y(t)\ge Y(0)+\eps]\le \exp\left(-\frac{\eps^2}{2tC^2}\right)$.

To that end, let $Y(t):=X(0)+\sum_{i=1}^{t} X(i)-\expec[X(i)\;\vert\;\cF_{i-1}]$.
Clearly, $Y(t)$ is adapted to $\cF_t$ and 
$\expec[Y(t+1)\;\vert\;\cF_t]=Y(t)$. Furthermore, we also have
\begin{align}
|Y(t+1)-Y(t)|&=\Big|X(t+1)-\expec[X(t+1)\;\vert\;\cF_t]\Big|
\le|X(t+1)-X(t)|+\Big|X(t)-\expec[X(t+1)\;\vert\;\cF_t]\Big|\le 2c_1\;.
\end{align}
Therefore, by Azuma, $\Pr[Y(t)\ge Y(0)+\eps]\le\exp\left(-\frac{\eps^2}{8tc_1^2}\right)$.

At the same time, let us see by induction that almost surely 
$X(t)\le Y(t)-c_2 t$ for every time $t$. Indeed $X(0)=Y(0)$ and then
\begin{align}
Y(t+1)&=Y(t)+X(t+1)-\expec[X(t+1)\;\vert\;\cF_t]
\overset{\text{ind.~hyp.~and}\eqref{eq:martingale_lemma_exp}}{\ge} X(t)+c_2t+X(t+1)-X(t)+c_2\\
&=X(t+1)+c_2(t+1)\;.
\end{align}
Therefore,
\begin{align}
\Pr[X(t)\ge X(0)-c_2t/2]
&\le\Pr[Y(t)-c_2t\ge Y(0)-c_2t/2]
\le\exp\left(-\frac{c_2^2}{32c_1^2}t\right)\;.\qedhere
\end{align}
\end{proof}

Before we turn to the proof of \Cref{thm:random-process},
let us quickly note that,
together with \Cref{lem:p0-p1-properties}, it implies
\Cref{thm:inactive-unstable-restated}:

\begin{proof}[Proof that \Cref{lem:p0-p1-properties} and
\Cref{thm:random-process} imply
\Cref{thm:inactive-unstable-restated}]
We set $\eps_1$ to be $\eps$ from \Cref{lem:p0-p1-properties}.
With that choice, $P_0(t)$ and $P_1(t)$ satisfy
\eqref{eq:p0-bounded}--\eqref{eq:p1-martingale}
if $P_0(t),P_1(t)>\Cmin=-\log\eps_1$.
Take $T$ from \Cref{lem:p0-p1-properties}
and
choose $\eps=\exp(-\Cstart)$ where $\Cstart$ is from
\Cref{thm:random-process}.

Let $\cU^{0}$ be $(\eps,\eps)$-inactive with
$P_0(0)<\infty$.
Then, $\min(P_0(0),P_1(0))>\Cstart$. Applying
\Cref{thm:random-process}, almost surely there exists
finite first time $t_0$ such that
$Q_0(t_0 T)=P_0(t_0)\le -\log\eps$
or $Q_1(t_0 T)=P_1(t_0)\le -\log\eps_1$.
Furthermore, by \Cref{lem:p0-p1-properties},
the configuration remains $(\epsbase,\epsbase)$-inactive until time
$t_0$. Finally, by~\eqref{eq:76},
with probability at least 0.7 it holds
$Q_1(t_0T)=P_1(t_0)>\Cmin$.
\end{proof}

\section{Proof of \Cref{prop:two-chains}}
\label{sec:stochastic-proof}
As a preliminary point, our assumption is that
\eqref{eq:p0-bounded}--\eqref{eq:p1-martingale} hold
whenever $P_0(t),P_1(t)>\Cmin$. In fact, let us assume
that these properties always hold. For example,
whenever the event 
$P_0(t)\le\Cmin$ or $P_1(t)\le\Cmin$ occurs,
we can redefine the random processes and
set them as $P_0(t+1)=P_0(t)-C$ and
$P_1(t+1)=P_1(t)+C$.
It should be clear that such a change does not affect the distributions
of $t_0$ and $P_1(t_0)$, so our
modification of $P_0$ and $P_1$ is without loss of generality.

First, we use a standard argument with Azuma inequality
to show that the stopping time $t_0$ is almost surely
finite, for any choice of $\Ctilde(C)\ge 0$.
Recall that $P_0(t)$ satisfies~\eqref{eq:p0-bounded}
and~\eqref{eq:p0-martingale}.
Therefore, applying \Cref{lem:modified-azuma},
it holds
$\Pr[P_0(t)\ge\Cstart-t/(2C)]\le \exp(-ct)$
for every $t$ and some fixed $c>0$.
However, if $P_0(t)>\Cstart$, then of course
$P_0(t)\ge \Cstart-t/(2C)$ for every $t$.
Hence,
\begin{align}
\Pr[t_0=\infty]
&\le \Pr[\forall t: P_0(t)> \Cstart]
\le \Pr[P_0(t)\ge \Cstart-t/(2C)\text{ infinitely often}]\\
&=\lim_{T\to\infty}
\Pr[\exists t\ge T:P_0(t)\ge \Cstart-t/(2C)]\\
&\le\lim_{T\to\infty} \sum_{t=T}^\infty
\exp(-ct)
=\lim_{T\to\infty} \frac{\exp(-cT)}{1-\exp(-c)}
=0\;.
\end{align}

\medskip

It remains to show that
$\Pr[P_1(t_0)\le\Cmin]\le 0.3$.
First, let us prove this statement with the assumption
$P_1(0)>\Cstart$ replaced with $P_1(0)>P_0(0)-C$.
For $\ell\ge 0$, let 
\begin{align}
    N_\ell&=\Big|\{t: \Cstart+\ell<P_0(t)\le \Cstart+\ell+1\}\Big|\;.
\end{align}
We are going to establish tail bounds on the values of $N_\ell$. Let 
$s_\ell$ be the first time such that $P_0(s_\ell)\le \Cstart+\ell+1$.
By \Cref{lem:modified-azuma} 
(which is applicable since $s_\ell$
is a stopping time, so $(P_0(s_\ell+t))_t$
is a random process satisfying
\eqref{eq:p0-bounded} and \eqref{eq:p0-martingale}),
for every $t\ge 0$
it holds
\begin{align}
\Pr[P_0(s_\ell+\lceil 2C\rceil+t)\ge \Cstart+\ell]
&\le 
\Pr\left[P_0(s_\ell+\lceil 2C\rceil+t)\ge P_0(s_\ell)-1\right]\\
&\le
\Pr\left[P_0(s_\ell+\lceil 2C\rceil+t)\ge P_0(s_\ell)-\frac{t+\lceil 2C\rceil}{2C}\right]\le\exp(-ct)\;.
\end{align}
That implies for a fixed $T\ge 0$
\begin{align}
\Pr\left[|N_\ell|>\lceil 2C\rceil +T\right]
&\le \Pr\left[\exists t\ge T:P_0(s_\ell+\lceil 2C\rceil+t)\ge\Cstart+\ell
\right]
\le\frac{\exp(-cT)}{1-\exp(-c)}\;.
\end{align}
For sufficiently large constant
$K'=K'(C)$, let us take
$T=K'\cdot(1+\ell)$.
Then, for every $\ell\ge 0$ it holds $\frac{\exp(-c T)}{1-\exp(-c)}\le \frac{0.1}{2^{\ell+1}}$.
Let $K=K'+\lceil 2C\rceil$. Then, by union bound,
\begin{align}
\Pr\left[\exists \ell\ge 0:|N_\ell|> K\cdot (1+\ell)\right]
&\le\Pr\left[\exists \ell\ge 0:|N_\ell|> \lceil 2C\rceil +K'\cdot (1+\ell)\right]
\le 0.1\;.
\end{align}
Hence, except with probability at most 0.1,
it holds $|N_\ell|\le K\cdot(1+\ell)$ for every $\ell\ge 0$.

\medskip

Assume that the event $P_1(t_0)\le\Cmin$ occurs. That is, there exists
some $t_0$ such that $P_1(t_0)\le \Cmin$
and $P_1(0),\ldots,P_1(t_0-1)>\Cmin$, and
$P_0(0),\ldots,P_0(t_0-1)>\Cstart$.
Then, by~\eqref{eq:p0-bounded}, it holds $P_0(t_0)>\Cstart-C$.
Consequently, $P_1(t_0)-P_0(t_0)<\Cmin-\Cstart+C=-\Ctilde+C\le -C$
if $\Ctilde$ satisfies $\Ctilde\ge 2C$.
Recall that we assumed $P_1(0)>P_0(0)-C$. 
Hence there exists the latest time $t'\le t_0$
such that $P_1(t'-1)>P_0(t'-1)-C$.
In particular, due to~\eqref{eq:p0-bounded}
and~\eqref{eq:p1-lower-bound} it holds 
$P_1(t')> P_0(t')-3C$.
Furthermore, by definition, 
$P_1(t'')\le P_0(t'')-C$ is satisfied for all times
$t'\le t''\le t_0$.

In light of this discussion, if $P_1(t_0)\le\Cmin$ occurs, then
there exist two times $t'\le t$ such that $P_1(t')> P_0(t')-3C$,
$P_1(t)\le\Cmin$,
and $P_1(t'')\le P_0(t'')-C$ for every
$t'\le t''\le t$.
For $\ell\ge 0$ and $i\ge 1$, let $T(\ell, i)$ be the $i$-th time step $t'$
such that $\Cstart+\ell<P_0(t')\le \Cstart+\ell+1$. Let $\cE(\ell, i)$ be the
event that, at the time $t'=T(\ell, i)$, we have
$P_1(t')> \Cstart+\ell-3C$, and that there exists $t\ge t'$ such that 
$P_1(t)\le\Cmin$ and $P_1(t'')\le P_0(t'')-C$
for all $t'\le t''\le t$.

By the discussion above, if $P_1(t_0)\le\Cmin$ occurs,
then either
there exists $\ell$ such that $N_\ell> K\cdot (1+\ell)$,
or there exist
$\ell$ and $1\le i\le K\cdot (1+\ell)$ such that $\cE(\ell, i)$ occurs.
In other words, by union bound we have
\begin{align}
\Pr[P_1(t_0)\le \Cmin]
&\le \Pr[\exists\ell\ge 0:|N_\ell|>K(1+\ell)]
+\sum_{\ell\ge 0}\sum_{i=1}^{K(1+\ell)}\Pr[\cE_{\ell,i}]\;.\\
&\le 0.1+\sum_{\ell\ge 0}\sum_{i=1}^{K(1+\ell)}\Pr[\cE_{\ell,i}]\;.
\end{align}

To estimate the probability of $\cE_{\ell,i}$, we use
the fact that $t'=T(\ell,i)$ is a stopping time
and apply \Cref{lem:modified-azuma}.
If $\Ctilde\ge 3C$, then $P_1(t')>\Cstart+\ell-3C>\Cmin$.
Since $P_1(t')> \Cstart+\ell-3C$ and $P_1(t'')\le P_0(t'')-C$ for $t''\ge t'$,
by~\eqref{eq:p1-bounded} it follows
$P_1(t'+s)>\Cstart+\ell-3C-Cs\ge \Cmin$, where the last inequality holds
as long as 
$s\le \frac{\ell}{C}+\frac{\Ctilde}{C}-3$.
Let $s_0=\lceil\frac{\ell}{C}+\frac{\Ctilde}{C}-3\rceil$.
It follows that
\begin{align}
\Pr[\cE_{\ell,i}]
&\le\sum_{s=0}^\infty
\Pr[P_1(t'+s)\le\Cmin\text{ and }
P_1(t'')\le P_0(t'')-C\text{ for }t'\le t''\le t+s]\\
&=\sum_{s\ge s_0}
\Pr[P_1(t'+s)\le\Cmin\text{ and }
P_1(t'')\le P_0(t'')-C\text{ for }t'\le t''\le t+s]\\
&\le\sum_{s\ge s_0}
\Pr[P_1(t'+s)\le P_1(t')\text{ and }
P_1(t'')\le P_0(t'')-C\text{ for }t'\le t''\le t+s]\\
&\overset{\text{Lem \ref{lem:modified-azuma}}}{\le}
\sum_{s=s_0}^\infty \exp(-cs)
\le\frac{\exp\left(-\frac{c\Ctilde}{C}+3c\right)}{1-\exp(-c)}
\exp\left(-\frac{c}{C}\ell\right)
\end{align}
for some constant $c(C)>0$ (note that $c$ does not depend on $\Ctilde$). 
Choosing sufficiently large $\Ctilde$, it follows
\begin{align}
\sum_{\ell=0}^\infty\sum_{i=1}^{K(\ell+1)}
\Pr[\cE_{\ell,i}]
&\le\frac{K\exp\left(-\frac{c\Ctilde}{C}+3c\right)}{1-\exp(-c)}
\sum_{\ell=0}^\infty(\ell+1)\exp\left(-\frac{c}{C}\ell\right)
=\frac{K\exp\left(-\frac{c\Ctilde}{C}+3c\right)}{1-\exp(-c)}
\frac{1}{\left(1-\exp\left(-\frac{c}{C}\right)\right)^2}
\le 0.1\;.
\end{align}

\medskip

To sum up, so far we showed that there exists a choice of $\Ctilde$
such that if
$P_0(0)>\Cstart$ and $P_1(0)>P_0(0)-C$, then
$\Pr[P_1(t_0)\le\Cmin]\le 0.2$.
In particular, the theorem is proved in the case of $P_0(0),P_1(0)>\Cstart$
and $P_1(0)>P_0(0)-C$. It remains to drop this last assumption.

To that end, assume that $P_0(0),P_1(0)>\Cstart$ and
$P_1(0)\le P_0(0)-C$.
Let $R(t)=P_1(t)-P_0(t)$.
As long as the condition
$P_1(t)\le P_0(t)-C$ holds, we have $|R(t+1)-R(t)|\le 2C$ and
$\expec[R(t+1)\;\vert\;\mathcal{F}_t]\ge R(t)+2/C$. Therefore, 
the stopping time $t_1=\min\{t: R(t)> -C\}$ is almost surely finite. 
It is sufficient to prove
\begin{align}
\label{eq:86}
\Pr[\exists t\le t_1: P_1(t)\le\Cmin]\le 0.1\;,
\end{align}
since if $P_1(t)>\Cmin$ for all $t\le t_1$, then either
$t_0\le t_1$, in which case certainly $P_1(t_0)>\Cmin$ or
$t_0>t_1$, in which case
$P_0(t_1)>\Cstart$ and $P_1(t_1)>P_0(t_1)-C$, so continuing
the process from $t_1$, by the first part of the proof, 
the event $P_1(t_0)\le \Cmin$ occurs with
additional probability of at most 0.2.
Accordingly, let us turn to showing~\eqref{eq:86} (for large enough $\Ctilde)$. 
\medskip

Due to~\eqref{eq:p0-bounded}
and~\eqref{eq:p1-lower-bound}, it holds
$P_1(t)>\Cmin=\Cstart-\Ctilde$ for $t\le\Ctilde/C$. On the other hand,
by \Cref{lem:modified-azuma}, it also holds
\begin{align}\label{eq:20}
\Pr[t\le t_1\text{ and }P_1(t)\le\Cmin]
&\le\Pr\left[t\le t_1\land P_1(t)\le P_1(0)\right]\le\exp(-ct)\;
\end{align}
for some $c(C)>0$. Let $t'$ be the smallest $t$ such that
$t>\Ctilde/C$. Then, as a consequence of~\eqref{eq:20}, it holds
\begin{align}
\Pr[\exists t\le t_1: P_1(t)\le\Cmin]
&\le\sum_{t=t'}^{\infty} \exp(-ct)
=\frac{\exp(-ct')}{1-\exp(-c)}\;.
\end{align}
It $\Ctilde$ is chosen large enough, then $t'$ satisfies $\frac{\exp(-ct')}{1-\exp(-c)}\le 0.1$ and indeed it follows
\begin{align}
\Pr[\exists t\le t_1: P_1(t)\le\Cmin]&\le 0.1\;,
\end{align}
which concludes the proof.\qed

\appendix
\section{Proof of \Cref{cor:bounded_step}}
\label{app:bounded-step}
Recall~\eqref{eq:new-aij}, which we will be using multiple times.
\begin{enumerate}
        \item
        If $|A_{ij}|< \eps_0$, then 
        $\min(|A_{i\ell}|,|A_{j\ell}|)< \eps_0$.
        Indeed, if $\min(|A_{i\ell}|,|A_{j\ell}|)>1-\eps_1^2$,
        then by \Cref{cl:consistent-signs}
        and since $\cU$
        is $(\eps_0,\eps_1)$-inactive,
        also
        $|A_{ij}|>1-\eps_1^2>\eps_0$, a contradiction.
        
        Therefore, from~\eqref{eq:new-aij},
        $|A'_{ij}|\le |A_{ij}|+\alpha|A_{i\ell}|\cdot|A_{j\ell}|
        <(1+\alpha)\eps_0\le\frac{1}{4(1+\alpha)}\le 1/2$.
        \item 
        If $|A_{ij}|> 1-\eps_1^2$ 
        and $\max(|A_{i\ell}|,|A_{j\ell}|)< \eps_0$, then
        \begin{align}
        |A'_{ij}|&>
        \frac{1-\eps_1^2}{\sqrt{1+(2\alpha+\alpha^2)\eps_0^2}}
        -\alpha\eps_0^2
        \ge (1-\eps_1^2)(1-(2\alpha+\alpha^2)\eps_0^2)-\alpha\eps_0^2
        \ge 1-(3\alpha+\alpha^2)\eps_0^2-\eps_1^2\\
        &\ge 1/2\;,
        \end{align}
        where the last line holds since from the assumption
        $\max(\eps_0,\eps_1^2)\le\frac{1}{4(2+\alpha)^2}$
        it follows
        $\eps_1^2\le 1/4$
        and 
        $(3\alpha+\alpha^2)\eps_0^2\le (3\alpha+\alpha^2)\eps_0\le 1/4$.
        By a similar calculation, it also holds
        $|A_{ij}-A'_{ij}|\le 1/2$, so
        $\sign(A_{ij})=\sign(A'_{ij})$.
        \item 
        If $|A_{ij}|> 1-\eps_1^2$ and
        $\max(|A_{i\ell}|,|A_{j\ell}|)>1-\eps_1^2$, then again
        by \Cref{cl:consistent-signs}
        it holds
        $\min(|A_{i\ell}|,|A_{j\ell}|)>1-\eps_1^2$and $\sign(A_{ij})=\sign(A_{i\ell})\sign(A_{j\ell})$. Then,
        \begin{align}
        |A'_{ij}|&>
        \frac{1-\eps_1^2+\alpha(1-\eps_1^2)^2}
        {1+\alpha}
        \ge
        1-\frac{1+2\alpha}{1+\alpha}\eps_1^2\ge\frac{1}{2}\;.
        \end{align}
\end{enumerate}
The lemma follows, as we exhausted
all possible cases.\qed

\section{Proof that \Cref{thm:technical-informal} implies \Cref{thm:main-polarization}}
\label{app:implication}

First, let us argue that an inactive configuration with one cluster polarizes.
This follows from a result proved
in~\cite{ABH24}.

\begin{lemma}[Lemma~3.11 in~\cite{ABH24}]
\label{lem:strictly-convex}
Let $\cU^0$ be an initial configuration of
$n$ agents such that
there exist $b_1,\ldots,b_n\in\{\pm 1\}$
with $\langle b_i\vec u_i^0, b_j\vec u_j^0
\rangle>0$ for every $i,j\in[n]$.
Then, $(\cU^t)_t$ polarizes almost surely.
\end{lemma}

\begin{corollary}
\label{cor:one-cluster}
Let $\cU^0$ be an $(\epsbase,\epsbase)$-inactive
initial configuration with one cluster.
Then, $(\cU^t)_t$ polarizes almost surely.
\end{corollary}

\begin{proof}
Let $\cU=\cU^0$ and $b_i=\sign(A_{1i})$.
Since $\cU$ has only one cluster, 
for every $i,j\in [n]$,
using \Cref{cl:consistent-signs},
it holds
$\sign(\langle b_i \vec u_i, b_j \vec u_j\rangle)
=\sign(A_{1i})\sign(A_{1j})\sign(A_{ij})
=1$. Therefore, $(\cU^t)_t$ polarizes
almost surely by \Cref{lem:strictly-convex}.
\end{proof}

Furthermore, we will use the fact that
a configuration which is not inactive
must become so. For that we need
an elementary geometrical claim:
\begin{claim}[Corollary 2.10 in \cite{ABH24}]
\label{cl:orthogonal-transitive}
Let $\eps>0$. If
$|A_{ij}|\le \eps$,
$|A_{ii'}|\ge 1-\eps^2$,
and $|A_{jj'}|\ge 1-\eps^2$,
then $|A_{i'j'}|\le 64 \eps$.
\end{claim}

\begin{lemma}
\label{lem:path-to-inactive}
For every $\eps>0$ there exists $K$
such that the following holds:
Let $\cU^0$ be any
initial configuration.
Then, there exists a sequence of $K$
interactions such that $\cU^K$ is
$(\eps,\eps)$-inactive.

In particular, almost surely, there exists
a time $t$ such that $\cU^t$ is $(\eps,\eps)$-inactive.
\end{lemma}

\begin{proof}
Let $\cU=\cU^0$ be a configuration.
Recall from~\eqref{eq:new-aij} that
if $0<|A_{i\ell}|<1$ and
agent $\ell$ influences agent $i$,
then their new correlation satisfies
$|A'_{i\ell}|>|A_{i\ell}|$. From this
and continuity, there exists $K_0=K_0(\eps)$
such that if 
$\eps/64\le |A_{i\ell}|\le 1-(\eps/64)^2$,
and $\ell$ influences $i$ for $K_0$ times,
then $|A_{i\ell}^{K_0}|>1-(\eps/64)^2$.

Let $K=K_0\cdot n$.
Let us define a sequence of at most $K$
interactions after which
the configuration is $(\eps,\eps)$-inactive.
(This sequence can be extended to length
exactly $K$, for example 
by adding interactions where
some $i$ influences itself.)

Let $S_1=\{i: |A_{1i}|\ge\eps/64\}$ and define the
\emph{anchor} of $S_1$ to be agent 1,
that is $w(1) = 1$.
For every agent $i\in S_1$,
let agent 1 influence agent $i$
for $K_0$ times. Let us call the new configuration $\Utilde$.
Consider $\Utilde$ with agents from
$S_1$ removed. If it is empty, stop.
Otherwise, apply the same procedure recursively
on the remaining agents. This results
in a new configuration with clusters
$S_2,\ldots,S_k$ and anchors
$w(2),\ldots,w(k)$. Let us add back $S_1$
and call the final configuration $\cU^K$.
Clearly, $\cU^K$ is constructed by applying
at most $K$ interactions to $\cU$.
Furthermore, we claim that $\cU^K$ is $(\eps,\eps)$-inactive with clusters
$S_1,\ldots,S_k$.

This is seen by induction on the number
of agents. In fact, let us prove
that the configuration $\cU^K$ is 
$(\eps,\eps)$-inactive,
and furthermore for every cluster $a$
and every $i\in S_a$ it holds
$|A^K_{w(a),i}|>1-(\eps/64)^2$.
Indeed, by induction, the clusters
$S_2,\ldots,S_k$ form an $(\eps,\eps)$-inactive
configuration. For $i,j\in S_1$,
by construction it holds
$\min(|A^K_{1i}|,|A^K_{1j}|)>1-(\eps/64)^2$,
which from \Cref{cl:consistent-signs}
implies $|A^K_{ij}|>1-(\eps/32)^2>1-\eps^2$.
Finally, for $i\in S_1,j\notin S_1$,
assume that $j\in S_a$ with the anchor
$w(a)$. By construction, agent $w(a)$
did not move and therefore we have
$|A^K_{1w(a)}|=|A_{1w(a)}|<\eps/64$.
Since also $|A^K_{1i}|>1-(\eps/64)^2$
and, by induction,
$|A^K_{j,w(a)}|>1-(\eps/64)^2$,
from \Cref{cl:orthogonal-transitive},
it follows $|A^K_{ij}|<\eps$.
\end{proof}

\begin{proof}[Proof of \Cref{thm:main-polarization}]
Let $\cU^0$ be a configuration which
is not separable. Recall constants
$\eps,\eps_1$ and $T$ from \Cref{thm:technical-informal}.

We define two sequences of stopping
times $\Tstart(\ell)$ and $\Tend(\ell)$, and a related sequence $\NC(\ell)$, as follows:
Let 
$\Tstart(0)=\min\{t: \cU^t \text{ is $(\eps,\eps)$-inactive}\}$. Note that $\Tstart(0)$ is almost surely
finite, by \Cref{lem:path-to-inactive}.

Given $\Tstart(\ell)$, let $\NC(\ell)$ be the number
of clusters of the configuration at time
$\Tstart(\ell)$. If $\NC(\ell)=1$,
let $\Tend(\ell)=\Tstart(\ell)$,
and $\Tend(\ell')=\Tstart(\ell')=\Tstart(\ell)$,
$\NC(\ell')=1$ for every $\ell'>\ell$.

If $\NC(\ell)>1$, then let
\begin{align}
\Tend(\ell)
&=\min\left\{
t: \text{
$t=\Tstart(\ell)+kT$ for some $k\ge 0$, and 
$\cU^t$ is not $(\eps,\eps_1)$-inactive
}
\right\}\;.
\end{align}
Since $\NC(\ell)>1$ and the configuration
is not separable, the assumptions of
\Cref{thm:technical-informal} are satisfied.
Hence, $\Tend(\ell)$ is almost surely finite. Finally, we let
\begin{align}
\Tstart(\ell+1)
&=\min\left\{t>\Tend(\ell): 
\text{$\cU^t$ is $(\eps,\eps)$-inactive}\right\}\;.
\end{align}
As at time $\Tend(\ell)$ the configuration
is not $(\eps,\eps_1)$-inactive, hence also
not $(\eps,\eps)$-inactive, the value
of $\Tstart(\ell+1)$ is almost surely finite
by \Cref{lem:path-to-inactive}.

\medskip

By \Cref{lem:cluster}, it holds $1\le \NC(\ell)\le d$
for every $\ell\ge 0$. We will now show that
almost surely there exists $\ell$ with $\NC(\ell)=1$.
By \Cref{cor:one-cluster}, that implies
that the process $(\cU^t)_t$ almost surely polarizes.

To that end, it is sufficient to show that
there exists a fixed $p>0$ such that
\begin{align}
\label{eq:25}
\Pr\left[\NC(\ell+1)\le \max(1,\NC(\ell)-1)
\;\vert\;\cU^{\Tstart(\ell)}\right]\ge p\;,
\end{align}
as indeed that implies
$\Pr[\NC(\ell+d)=1\;\vert\;\cU^{\Tstart(\ell)}]
\ge p^d$ and therefore $\NC(\ell)=1$ almost surely
happens for some $\ell$.

To show~\eqref{eq:25}, consider a 
configuration $\cU$ at time $\Tstart(\ell)$
such that $\NC(\ell)>1$. By \Cref{thm:technical-informal}, with probability at least 0.7,
the configuration $\Utilde$ at time
$\Tend(\ell)$ is $(\epsbase,\epsbase)$-inactive with clusters
$S_1,\ldots,S_{\NC(\ell)}$ and furthermore 
has two distinct clusters $S_a,S_b$ and
opinions $i_0\in S_a,j_0\in S_b$ such that
$|\Atilde_{i_0j_0}|\ge \eps$.
Given such $\Utilde$, we will now define a sequence
of at most $K$ (for some fixed $K$) 
interactions such that the resulting
configuration $\Uhat$ is $(\eps,\eps)$-inactive,
with at most $\NC(\ell)-1$ clusters. That implies
$\Pr[\NC(\ell+1)\le \NC(\ell)-1]\ge 0.7\cdot n^{-2K}$,
and therefore \eqref{eq:25}, concluding the proof.

\medskip

First, for every $i\in S_a$ such that
$|\Atilde_{ij_0}|<\frac{\alpha\eps}{4}$,
let $i_0$ influence $i$ one time. 
Let this intermediate configuration be called $\Utilde'$.
After any such interaction,
from~\eqref{eq:new-aij} and due to
$|\alpha \Atilde_{ii_0}\Atilde_{i_0j_0}|\ge\frac{\alpha\eps}{2}$, it holds
$|\Atilde'_{ij_0}|\ge\frac{\alpha\eps}{4(1+\alpha)}$.
Therefore, we obtain a configuration where
for every $i\in S_a$ it holds
\begin{align}
\label{eq:26}
|\Atilde'_{ij_0}|&\ge 
\min\left(\frac{\alpha\eps}{4},\frac{\alpha\eps}{4(1+\alpha)}\right)=\frac{\alpha\eps}{4(1+\alpha)}\;.
\end{align}
Let $\eps'=\min\left(\frac{\alpha\eps}{4(1+\alpha)},\frac{\eps}{64}\right)$ and
\begin{align}
S=\left\{i: |\Atilde'_{ij_0}|\ge\eps'\right\}\;.
\end{align}
There exists a fixed $K_0$ such that if agent $j_0$
influences $i\in S$ for $K_0$ times, then their 
new absolute correlation exceeds $1-(\eps/64)^2$.
Let $j_0$ influence every $i\in S$ for $K_0$ times.
Note that $S_a\cup S_b\subseteq S$, where for $i\in S_a$
this follows from~\eqref{eq:26} and for $i\in S_b$
since $|\Atilde'_{ij_0}|=|\Atilde_{ij_0}|> 1-\epsbase^2$.

After that, forget about the agents in $S$ and apply
the procedure from the proof of \Cref{lem:path-to-inactive}
to the remaining agents. Call the final configuration $\Uhat$.
Indeed, this configuration is obtained from $\Utilde$ using
$O(1)$ interactions. From \Cref{lem:path-to-inactive},
configuration $\Uhat$ is $(\eps,\eps)$-inactive, 
with clusters
$S,\Shat_2,\ldots,\allowbreak\Shat_k$
and anchors $j_0,\what(2),\allowbreak\ldots,\allowbreak\what(k)$.
And indeed $k<\NC(\ell)$, since, as already mentioned,
$S_a\cup S_b\subseteq S$, and on the other hand
for any distinct $a',b'$ the anchors $\what(a')$ and $\what(b')$
could not have been in the same cluster in $\Utilde$:
On the one hand, the anchors have the same position
in $\Utilde$, and if they were in the same cluster
their absolute correlation must be more than $1-\epsbase^2$.
On the other hand, by construction,
their mutual absolute
correlations must be at most $\eps/64$, a contradiction.
\end{proof}

\printbibliography
\end{document}